\documentclass{article}

\usepackage{amssymb, amsmath, amsthm, graphicx, mathtools, url, enumitem, balance, hyperref}
\usepackage{tikz, pgfplots, filecontents, multirow}
\usepackage{algorithm, algpseudocode}
\usepackage{subcaption}
\usepackage[inkscapeformat=png]{svg}
\usepackage[labelformat=simple]{subcaption}
\usepackage{xcolor}
\usepackage{multirow, booktabs}
\usepackage{threeparttable}

\usepackage{pdfpages, fullpage, textcomp}

\usetikzlibrary{patterns, positioning, arrows.meta, shapes, positioning, fit, backgrounds, ext.paths.ortho, trees}

\definecolor{OI1}{RGB}{230,159,0}
\definecolor{OI2}{RGB}{86,180,233}
\definecolor{OI3}{RGB}{0,158,115}
\definecolor{OI4}{RGB}{240,228,66}
\definecolor{OI5}{RGB}{204,121,167}
\definecolor{OI6}{RGB}{213,94,0}
\definecolor{lightergray}{RGB}{210, 210, 210}

\newtheorem{thm}{Theorem}
\newtheorem{lem}{Lemma}
\newtheorem{rmk}{Remark}
\newtheorem{pro}{Problem}
\newtheorem{col}{Corollary}
\newtheorem{mydef}{Definition}
\newtheorem{eg}{Example}
\newtheorem{hyp}{Hypothesis}

\title{Generalizing Fair Top-$k$ Selection: An Integrative Approach}
\author{
    Guangya Cai\footnote{University of Minnesota, Twin Cities, MN, USA; \texttt{cai00171@umn.edu}.}
}

\date{\vspace{-5ex}}

\begin{document}

\maketitle

\begin{abstract}
  Fair top-$k$ selection, which ensures appropriate proportional representation of members from minority or historically disadvantaged groups among the top-$k$ selected candidates, has drawn significant attention. We study the problem of finding a fair (linear) scoring function with multiple protected groups while also minimizing the disparity from a reference scoring function. This generalizes the prior setup, which was restricted to the single-group setting without disparity minimization. Previous studies imply that the number of protected groups may have a limited impact on the runtime efficiency. However, driven by the need for experimental exploration, we find that this implication overlooks a critical issue that may affect the fairness of the outcome. Once this issue is properly considered, our hardness analysis shows that the problem may become computationally intractable even for a two-dimensional dataset and small values of $k$. However, our analysis also reveals a gap in the hardness barrier, enabling us to recover the efficiency for the case of small $k$ when the number of protected groups is sufficiently small. Furthermore, beyond measuring disparity as the ``distance'' between the fair and the reference scoring functions, we introduce an alternative disparity measure---utility loss---that may yield a more stable scoring function under small weight perturbations.  Through careful engineering trade-offs that balance implementation complexity, robustness, and performance, our augmented two-pronged solution demonstrates strong empirical performance on real-world datasets, with experimental observations also informing algorithm design and implementation decisions.
\end{abstract}

\section{Introduction}
Top-$k$ selection, where one selects $k$ most relevant items from a dataset with $n$ items, is a fundamental task in information retrieval and decision-making. To indicate the relevance of an item, each item is typically assigned a score derived from its attributes. Then, the $k$ highest-scoring items are selected. For example, in school admission, a college admissions officer may evaluate several attributes of an applicant (GPA, SAT score, essay score, etc.), assign an overall score to each applicant, and admit the (say) top-500 students with the highest scores. To determine such a score, a scoring function is needed, which is typically a linear scoring function that weights and sums the attribute values by relative importance.

With decision-making increasingly being aided by automatic algorithms in many areas, such as hiring \cite{geyik2019fairness} and school admission \cite{peskun2007effectiveness}, there has been growing concern that using automated algorithms may result in decreased fairness, thereby promoting stereotypes and polarizing opinions. Given the importance of top-$k$ selection in the algorithmic decision-making system, fair top-$k$ selection has drawn significant attention~\cite{zehlike2017fa, kleinberg2018selection, celis2018ranking, asudeh2019designing, yang2019balanced, campbell2024query, liu2024fair, cai2025finding}. Among many different ways of quantifying fairness (e.g.~\cite{dwork2012fairness, yang2017measuring, friedler2021possibility, zehlike2022fairness}), a common one is based on proportional representation. Each item has categorical attributes (e.g., gender, race, ethnicity), and one or a combination of these attributes is regarded as a \textit{sensitive} attribute. Some values of this sensitive attribute indicate memberships of minority or historical disadvantaged groups, often referred to as~\textit{protected} groups. The goal of fair top-$k$ selection is to select a top-$k$ subset where for each protected group, its proportional in the top-$k$ subset is roughly the same as that in the whole dataset.

In many applications, one often want the selected top-$k$ subset being both high-quality and fair. This is often achieved by applying proportional fairness constraints. For each protected group, one specifies a lower bound and/or upper bound for its proportion in the selected top-$k$ subset. In many previous works (e.g.,~\cite{zehlike2017fa, celis2018ranking, yang2019balanced, campbell2024query}), the fairness constraints are applied after scores to each item are identified. However, this method may result in a different selection criterion to different groups of people, introducing potential legal risks~(e.g., Ricci vs. DeStefano~\cite{ricci}). So, designing a fair scoring function to begin with could be a better choice. Consider the following variation of the example
from~\cite{asudeh2019designing, cai2025finding}.

\begin{eg}\label{eg:admission}
   A college admissions officer is designing a selection process, let us say, based on applicants' GPA ($g$) and SAT score ($s$), both normalized. The officer might believe that $g$ and $s$ should have an approximately equal weight, thus having the scoring function $f(c) = 0.5 \times g + 0.5 \times s$ to compute a score for each applicant $c$. Then, the (say) top-500 applicants are selected by score. However, this may yield only 150 female and 30 Black applicants in the top-500 whereas the fairness constraints are that at least $40\%$ and $10\%$ of the admitted class be female and Black, respectively. This violation may be due to gender and race disparities in the SAT score~\cite{sat}. 
    
    To address this, the officer may try to find a ``nearby'' fair scoring function that preserves the explainability of the choice of weights (approximately equal weighting). The weights of $g$ and $s$ can range between $0.45$ and $0.55$. Since multiple possible fair scoring functions may exist within this range, the officer may use a secondary criterion to choose among them, for example, choosing the one that minimizes the sum of absolute differences from the original weights. By running an algorithm for finding such a scoring function, the officer might successfully find a fair scoring function $f(c) = 0.54 \times g + 0.46 \times s$.
    
    However, it is possible that no such function within the allowable weight range exists. For example, adjusting the weights by at most $0.05$ may be sufficient to find a scoring function to satisfy the gender fairness constraint but still fail to satisfy the racial fairness constraint. In such cases, the algorithm reports a failure. This suggests that the officer's design principle (assigning approximately equal weights) may not be appropriate for the fairness consideration. The officer may then revise the weighting principle and with the aid of the algorithm for finding a fair scoring function, arrive at a new scoring function $f(c) = 0.6\times g + 0.4 \times s$ that satisfies both the gender and racial fairness constraints.
\end{eg}

Previous work~\cite{asudeh2019designing, liu2024fair, cai2025finding} has explored finding such a fair scoring function in a similar setup. However, most of these methods~\cite{asudeh2019designing, liu2024fair} do not scale efficiently for high-dimensional datasets and do not address a critical issue: potential ties among candidates. Since tie-breaking can affect the number of selected candidates from protected groups, this issue merits greater attention~\cite{cachel2025group,cai2025finding}. A more recent approach~\cite{cai2025finding} improves scalability and explicitly identifies and addresses the tie-breaking issue; however, it considers only a single protected group and returns an arbitrary fair weight vector within the allowable scoring function subspace, which may not be desirable in some applications (as illustrated in Example~\ref{eg:admission}). 

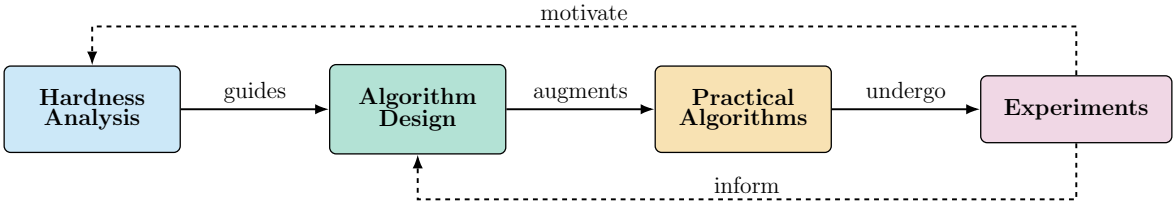
\begin{figure}[tbh!]
    \centering
    \resizebox{0.95\textwidth}{!}{
    \begin{tikzpicture}[
  cluster/.style = {draw, thick, fill=#1!30, inner sep=14pt, rounded corners},
  node distance=12pt,
  level distance=120pt
]

\node[inner sep=0pt, text width=80pt, align=center] (hard) {\Large\textbf{Hardness\\Analysis}};
\begin{scope}[on background layer]
    \node[cluster=OI2, fit=(hard), anchor=center] (hardness_cluster) {};
\end{scope}

\node[inner sep=0pt, right=120pt of hard, text width=80pt, align=center] (alg) {\Large\textbf{Algorithm\\Design}};
\begin{scope}[on background layer]
\node[cluster=OI3, fit=(alg), anchor=center] (alg_cluster) {};
\end{scope}

\node[inner sep=0pt, right=120pt of alg, text width=80pt, align=center] (pract) {\Large \textbf{Practical Algorithms}};
\begin{scope}[on background layer]
\node[cluster=OI1, fit=(pract)] (pract_cluster){};
\end{scope}

\node[inner sep=0pt, right=120pt of pract] (exp) {\Large\textbf{Experiments}};
\begin{scope}[on background layer]
\node[cluster=OI5, fit=(exp)] (exp_cluster) {};
\end{scope}

\draw [very thick, -{Latex[length=3mm]}] (hardness_cluster.east) -- (alg_cluster.west) node[midway,above] {\Large guides};
\draw [very thick, -{Latex[length=3mm]}] (alg_cluster.east) -- (pract_cluster.west) node[midway,above] {\Large augments};
\draw [very thick, -{Latex[length=3mm]}] (pract_cluster.east) -- (exp_cluster.west) node[midway,above] {\Large undergo};
\draw [very thick, -{Latex[length=3mm]}, dashed] (exp_cluster.north) |-| [ratio=5.0] (hardness_cluster.north) node[midway, above] {\Large motivate};
\draw [very thick, -{Latex[length=3mm]}, dashed] (exp_cluster.south) |-| [ratio=5.0] (alg_cluster.south) node[midway, above] {\Large inform};
\end{tikzpicture}
}
    \caption{The structure and interplay of key components in this work. Solid arrows denote the primary workflow, while dashed arrows indicate feedback and motivation.}\label{fig:structure2}
\end{figure}

\subsection{Our contributions}
In this work, we study the problem of finding a fair linear scoring function with multiple protected groups, while also minimizing the disparity from a reference, unfair scoring function.  We adopt an integrative framework similar to that of \cite{cai2025finding}, encompassing both high-level theoretical reasoning and low-level engineering optimizations (see Figure~\ref{fig:structure2}). Previous work~\cite{asudeh2019designing, liu2024fair} implies that a larger number of protected groups may not impose significant runtime overhead. However, we find that this implication overlooks the tie-breaking issue identified in~\cite{cai2025finding}. Once this critical issue that may affect the fairness of the outcome is properly considered, we show that the problem becomes NP-hard even for a two-dimensional dataset. Furthermore, the ``Small $k$ Opportunity'' identified in~\cite{cai2025finding}, which yields an efficient algorithm when $k$ is sufficiently small, largely vanishes in the worst case: we establish a near $\Omega(n^k)$ lower bound for constant values of $k$ under well-known conjectures in fine-grained complexity~\cite{williams2018some}. Notably, the hardness analysis is driven by the need for generating inputs for our experimental explorations. Despite these negative findings, the hardness analysis also reveals a gap in the hardness barrier, that when the number of protected groups is sufficiently small, the ``Small $k$ Opportunity'' can be recovered. With this gap exploited, we further augment the algorithms of the two-pronged solution proposed in~\cite{cai2025finding}, making them efficient in finding a fair scoring function with multiple protected groups while minimizing the disparity. Specifically, we introduce a new disparity measure---utility loss---which differs from those used in prior works~\cite{asudeh2019designing, liu2024fair} and enables us to find a more stable scoring function under small weight perturbations. To implement our augmented two-pronged solution in practice, we address several practical engineering considerations, balancing implementation complexity, robustness, and performance through careful engineering trade-offs. Experimental results on real-world datasets demonstrate the efficiency of our augmented two-pronged solution, while also informing algorithm design and implementation decisions, including the choice between the two algorithms in our augmented two-pronged solution.

The source code, data, and other artifacts have been made available at \url{https://github.com/caiguangya/fair-topk-general}.

\subsection{Related work}
Algorithmic fairness has emerged as a prominent research area in computer science, with particular relevance to data processing systems. A broad range of research has addressed fairness across diverse tasks, including clustering~\cite{chierichetti2017fair, chen2019proportionally}, dimensionality reduction~\cite{samadi2018price, kleindessner2023efficient}, and matching~\cite{chierichetti2019matroids, esmaeili2023rawlsian}. Fair top-$k$ selection, especially the approach that selects the $k$-best items under proportional fairness constraints, has also been widely studied \cite{zehlike2017fa, kleinberg2018selection, celis2018ranking, asudeh2019designing, yang2019balanced, campbell2024query, cai2025finding}. See \cite{zehlike2022fairness} for a survey. Notably, intersectional fairness, which considers candidates belonging to multiple different minority or historically disadvantaged groups, has also drawn considerable attention \cite{yang2021causal, kong2022intersectionally, gohar2023survey},  including in the context of fair top-$k$ selection~\cite{yang2021causal}.

For the problem of fair top-$k$ selection, most prior work either scores items on a single attribute or precomputes scores from multiple attributes using a fixed, predefined scoring function.  In many applications, however, selecting the top-$k$ items requires considering several attributes simultaneously, and choosing an appropriate scoring function over these attributes could be a principled way to promote fairness. The authors in \cite{asudeh2019designing} proposed algorithms to find a fair linear scoring function if the initial one does not produce top-$k$ subsets meeting the given fairness constraints. Since a linear scoring function can be presented as a weight vector, they aim at searching for a fair weight vector minimizing the angular distance. Their algorithms can work with various fairness models based on proportional fairness constraints, including proportional fairness constraints considered in this work. However, their algorithms do not scale efficiently with the size of dataset, and their run times increase dramatically with dimensionality. Moreover, the secondary criteria considered in the work, angular distance, requires solving expensive non-linear programs. Besides applying proportional fairness constraints, the authors in \cite{liu2024fair} introduced a fairness metric, alpha-fairness, which is based on the differences between the proportions of protected group members in the top-$k$ subset and their proportions in the entire dataset. They further proposed methods for finding a fair weight vector optimizing this metric. However, their approach lacks a formal runtime analysis and still scales poorly. Moreover, the secondary criteria considered in this work is difficult to interpret, as their weight vectors are normalized by the $L_1$ norm while the distance between the fair and the reference weight vectors is measured by their $L_2$ distance. In particular, neither work~\cite{asudeh2019designing, liu2024fair} properly addresses the potential ties among candidates, and both use a distance metric to measure the disparity, which produces a scoring function that is inherently unstable (see Section~\ref{subsec:theoretical-klevel}). The author in \cite{cai2025finding} adopted an integrative framework to systematically study the problem of finding a fair weight vector under fairness constraints, resulting in a more efficient two-pronged solution while also offering valuable insights into the problem structure and algorithmic behavior. However, they only considered a single protected group and did not use a secondary criterion for choosing among multiple possible fair scoring functions.

\section{Preliminaries and problem definition}\label{sec:prelim}
In this problem, we are given a set, $\mathcal{C}$, of candidates (or items), each described by $d$ scoring attributes and a set of sensitive attribute values. For a candidate $c\in \mathcal{C}$, we use a point $p(c) \in \mathbb{R}^d$ to represent its scoring attribute values and an appropriate subset of labels $A(c) \subseteq \{\mathcal{G}_1,\mathcal{G}_2,\dots,\mathcal{G}_m\}$ to represents its sensitive attribute values. Among these labels, we have $\mathcal{G}_j$ for all $1 \leq j \leq n_p$ to be protected groups, with $n_p$ being the total number of protected groups. Following~\cite{zehlike2022fairness}, we also use $\mathcal{G}_j$ to denote the subset of candidates with $\mathcal{G}_j \in A(c)$, when the context is clear. Our definition of the sensitive attribute is slightly different from that in \cite{zehlike2022fairness}, where protected groups may be values of different sensitive attributes. Our formulation can be viewed as concatenating all relevant attributes into a single composite sensitive attribute. Also, one can construct ``artificial'' sensitive attribute values corresponding to the intersections of multiple protected groups, which makes it convenient to apply intersectional fairness constraints. An illustrative example is given below.

\begin{eg}\label{eg:attribute}
Consider a dataset with gender and race information for each candidate. The sensitive attribute concatenates the gender and race together, such that for each candidate $c$, its sensitive attribute value, $A(c)$, is an appropriate subset of $\{\text{Male},\text{Female},\ldots,\allowbreak \text{Caucasian}, \allowbreak \text{Black},\ldots \}$. Thus, the sensitive attribute value of a black female candidate is $A(c) = \{\text{Female},\text{Black}\}$. Furthermore, if both ``Female'' and ``Black'' are protected groups, an ``artificial'' sensitive attribute value, ``Black Female'' can be constructed and added to the attribute domain. In this case, the sensitive attribute value of a black female candidate becomes $A(c) = \{\text{Female},\text{Black},\allowbreak \text{Black Female}\}$.
\end{eg}

We focus on the same class of linear scoring functions as the previous work~\cite{cai2025finding}. Such a function is given by a real-valued weight vector $w = (w_1, w_2, \dots, w_d)$, where each $w_i\geq 0$ and $||w||_1 = 1$. Following~\cite{cai2025finding}, a top-$k$ subset $\tau_k^w$ of a given value $k$ and a weight vector $w$ is defined as follows:

\begin{mydef}
For a given weight vector $w$ and a positive integer $k \leq n$, a subset $\tau_k^w \subseteq C$ is a top-$k$ subset of $C$ if $|\tau_k^w| = k$ and for any pair of candidates that $c \in \tau_k^w$ and $c' \in C \setminus \tau_k^w$, $w \cdot p(c) \geq w \cdot p(c')$.
\end{mydef}

Since the top-$k$ subset for a $w$ may not be unique due to ties among candidates, we consider all top-$k$ subsets and regard $w$ as fair if any of its top-$k$ subsets satisfies all fairness constraints. For each protected group $\mathcal{G}_j$, the corresponding lower bound and upper bound are given as $L_{k}^{\scriptscriptstyle \mathcal{G}_i}$ and $U_{k}^{\scriptscriptstyle \mathcal{G}_j}$. Vector $w$ is a fair scoring weight vector if 
\begin{equation}
    L_{k}^{\scriptscriptstyle \mathcal{G}_j} \leq | \tau_k^w \cap \mathcal{G}_j | \leq U_{k}^{\scriptscriptstyle \mathcal{G}_j},
\end{equation}
for all $1 \leq j \leq n_p$. One thing to note is that since ``artificial'' sensitive attribute values are introduced, independent fairness constraints can be applied to the intersection of multiple protected groups. In the case of Example~\ref{eg:attribute}, one can apply a specific fairness constraint to ``Black Female'', distinct from those applied to ``Female'' and ``Black'' individually, as applying fairness constraints to ``Female'' and ``Black'' may not guarantee an adequate representation of ``Black Female'' candidates~\cite{yang2021causal}.

Similar to the previous problem formulation in~\cite{cai2025finding}, we also aim to find the fair weight vector in a subset $V$ of the weight vector space described by $l$ linear inequalities. For a given $V$, there may be several fair weight vectors in $V$ that satisfies the fairness constraints. To identify the most suitable one, we are also given an input unfair weight vector $w^o$ as a reference weight vector, and we aim to find a fair weight vector $w^{\scriptscriptstyle f}$, with the minimum disparity from $w^o$ for one of the two following objectives:

\begin{itemize}
    \item \textbf{$\boldsymbol{w}$ difference.} The disparity of a fair weight vector $w^{\scriptscriptstyle f}$ is measured by the $L_1$ distance between it and the given reference weight vector $w^o$. Formally, the objective $g^o$ is defined as
    \begin{equation}
        g^o(w^{\scriptscriptstyle f}) = ||w^{\scriptscriptstyle f} - w^o||_1 = \sum_{i=1}^d |w^{{\scriptscriptstyle f}}_i - w^o_i|.
    \end{equation}
    This objective can be regarded as the sum of the absolute differences between corresponding components of the weight vectors.
    \item \textbf{Utility loss.} By the principle in \cite{zehlike2022fairness}, the utility of a top-$k$ subset $\tau_k$ is computed by summing scores of all candidates in $\tau_k$ under the reference weight vector $w^o$. That is
    \begin{equation}
        U^o(\tau_k) = \sum_{c\in \tau_k} w^o \cdot p(c).
    \end{equation}
    The disparity of a fair weight vector $w^{\scriptscriptstyle f}$ is measured by the relative utility loss of its corresponding top-$k$ subset. Formally, the objective $g^o$ is defined as
     \begin{equation}
        g^o(w^{\scriptscriptstyle f}) = 1 - \frac{U^o(\tau^{\scriptscriptstyle f}_k)}{U^o(\tau^o_k)} =  1 - \frac{\sum_{c\in \tau^{\scriptscriptstyle f}_k} w^o \cdot p(c)}{\sum_{c\in \tau_k^o} w^o \cdot p(c)},
    \end{equation}
    where $\tau^{\scriptscriptstyle f}_k$ and $\tau^o_k$ are top-$k$ subsets of $w^{\scriptscriptstyle f}$ and $w^o$, respectively.
\end{itemize}
Note that for the utility loss objective, there may be multiple top-$k$ subsets for a $w^{\scriptscriptstyle f}$ due to ties in scores. The top-$k$ subset that satisfies the fairness constraints and minimizes the relative utility loss is used to determine the value of $g^o(w^{\scriptscriptstyle f})$. In fact, minimizing the utility loss is equivalent to maximizing the utility (under $w^o$), as $U^o(\tau^o_k)$ stays the same regardless of the tie-breaking outcome. These observations are illustrated in the following example.

\begin{eg}\label{eg:utility}
Consider a dataset of $5$ candidates with scoring attribute values being $\{(0.4, 0.7), (0.5, 0.6), \allowbreak (0.7, 0.35), (0.8, 0.2), (0.9, 0.9) \}$ (sensitive attribute values are omitted). Let $k = 2$ and suppose that the given reference weight vector is $w^o=(0.5, 0.5)$. Tie-breaking leads to two top-2 subsets: $\{(0.9, 0.9), (0.4, 0.7)\}$ and $\{(0.9, 0.9), (0.5, 0.6)\}$. Both top-2 subsets have the same utility $1.45$ under $w^o$. Now, suppose we obtain a fair vector $w^{\scriptscriptstyle f} = (0.6, 0.4)$, then we have two possible top-2 subsets due to ties: $\{(0.9, 0.9), (0.7, 0.35)\}$ and $\{(0.9, 0.9), (0.8, 0.2)\}$. Their utilities under $w^o$ are $1.425$ and $1.4$ respectively. Assuming both satisfy the fairness constraints, the former subset is used to determine the utility loss of $w^{\scriptscriptstyle f}$.
\end{eg}
One benefit of the utility loss objective is that it enables improvement in the stability of the resulting fair scoring function, which will be discussed in Section~\ref{subsec:theoretical-klevel}.

Formally, we define our problem as follows:

\begin{pro}[(Generalized) Fair Top-$k$ Selection]
    We are given a set of $n$ candidates with $d$ scoring attributes, one sensitive attribute, and $n_p$ protected groups. Each candidate, $c$, is represented by a point $p(c)$ in $\mathbb{R}^d$ and a subset $A(c) \subseteq \{\mathcal{G}_1,\mathcal{G}_2,\dots,\mathcal{G}_m\}$. We are also given a non-negative integer $k \leq n$, two sets of non-negative integers $\{ L_{k}^{\scriptscriptstyle \mathcal{G}_j} \mid 1\leq j \leq n_p \}$ and  $\{ U_{k}^{\scriptscriptstyle \mathcal{G}_j} \mid 1\leq j \leq n_p \}$, a subset $V$ of the weight vector space described by $l$ linear inequalities, and a reference weight vector $w^o$ with its corresponding a disparity measure $g^o: V \to \mathbb{R}$ ($w$ difference or utility loss). The goal is to find a real-valued weight vector $w =(w_1, w_2, \ldots, w_d) \in V$, where $w_i\geq 0$ for all $i$ and $||w||_1 = 1$, such that
    \begin{equation*}
        L_{k}^{\scriptscriptstyle \mathcal{G}_j} \leq | \tau_k^w \cap \mathcal{G}_j | \leq U_{k}^{\scriptscriptstyle \mathcal{G}_j},
    \end{equation*}
    for all $1 \leq j \leq n_p$, while also minimizing the value of $g^o(w)$.
\end{pro}

\section{Hardness}
In a previous work~\cite{cai2025finding}, the author considered the hardness of a special version of the problem formulated above, where $n_p = 1$ and the disparity minimization was not required. Those results also apply here: the problem is 3SUM-hard for $d\geq 3$, has a lower bound of $\Omega(n^{d-1})$ for $d\geq 3$ under a certain decision model, and is NP-hard in arbitrary dimensions. However, the analysis also revealed the following gaps in the hardness barrier:

\begin{itemize}
    \item \textbf{Low Dimensions Opportunity.} The problem is NP-hard in arbitrary dimensions but solvable in polynomial time for fixed dimensions, with upper and lower bounds becoming smaller as the dimensionality $d$ decreases. Notably, the previous hardness analysis does not cover the $d = 2$ case. 
    \item \textbf{Small $k$ Opportunity.} The hardness analysis does not apply to a sufficiently small $k$. Subsequent algorithmic results confirm that the lower bounds are indeed breakable in such cases (e.g. $k = O(\textnormal{\text{polylog}}(n))$).
\end{itemize}
Previous work~\cite{asudeh2019designing, liu2024fair} does not examine how the number of protected groups affects runtime efficiency, and implies that a larger number of protected groups may have only a limited impact. In this section, however, we show that with multiple protected groups, the above ``opportunities'' in the single-protected-group setting may vanish in the worst case. For the ``Low Dimensions Opportunity'', we show that even for $d=2$, the problem is NP-hard for a large $n_p$. For the ``Small $k$ Opportunity'', we show that for any constant $k\geq 2$, the problem has a (conditional) lower bound of $\Omega(n^{k - \delta})$ for any constant $\delta > 0$, even when $n_p = \alpha \log n$ where $\alpha$ is a constant. These results are established using the problem of verifying a fair weight vector, described as bellow.

\begin{pro}[Fair Top-$k$ Verification]
    We are given a set of $n$ candidates with $d$ scoring attributes, one sensitive attribute and $n_p$ protected groups. Each candidate, $c$, is represented by a point $p$ in $\mathbb{R}^d$ and a subset $A(c) \subseteq \{\mathcal{G}_1,\mathcal{G}_2,\dots,\mathcal{G}_m\}$. We are also given a non-negative integer $k \leq n$, two sets of non-negative integers $\{ L_{k}^{\scriptscriptstyle \mathcal{G}_j} \mid 1\leq j \leq n_p \}$ and  $\{ U_{k}^{\scriptscriptstyle \mathcal{G}_j} \mid 1\leq j \leq n_p \}$, and a weight vector $w$. The goal is to determine whether there exists a top-$k$ subset $\tau_k^w$ of $w$ such that
    \begin{equation*}
        L_{k}^{\scriptscriptstyle \mathcal{G}_j} \leq | \tau_k^w \cap \mathcal{G}_j | \leq U_{k}^{\scriptscriptstyle \mathcal{G}_j},
    \end{equation*}
    for all $1 \leq j \leq n_p$.
\end{pro}

At first glance, the problem appears trivial: one can use $w$ to compute a score for each candidate, apply a selection algorithm to identify the top-$k$ candidates, and then count the number of selected candidates of each protected group to verify whether the fairness constraints are satisfied. However, as noted in Example~\ref{eg:utility}, the top-$k$ subset may not be unique due to ties in scores. In such cases, it is no longer obvious how to efficiently determine whether there exists a top-$k$ subset satisfying the fairness constraints. The problem is computationally easier than the Fair Top-$k$ Selection problem, as an algorithm to the latter problem can be used to solve the former. In fact, this problem can be regarded as the Fair Top-$k$ Selection problem when $d=1$, treating the value $w\cdot p(c)$ as the single scoring attribute value. While the specific weight assigned to this attribute becomes irrelevant, the algorithm can effectively determine whether there exists a top-$k$ subset that satisfies the fairness constraints. Consequently, any lower bounds established for the Fair Top-$k$ Verification problem can be applied to the Fair Top-$k$ Selection problem.

Notably, the Fair Top-$k$ Verification problem also plays a crucial role in our experimental evaluation, specifically for input generation. Since we sample weight vectors uniformly at random until a target number of \textit{unfair} weight vectors are obtained, an efficient algorithm to determine whether a sampled weight vector is fair or not is essential for our experimental explorations (see Section~\ref{subsec:exp-setup}).

\subsection{Hardness result in low dimensions}
In this section, we consider the hardness of the problem in low dimensions, specifically for $d=2$. We establish the NP-hardness of the problem for an arbitrary $n_p$; that is, the number of protected groups can be arbitrarily large, as long as $n_p = O(n)$. We use a classical NP-hard problem, Set Cover, for the reduction, and consider an instance of the Fair Top-$k$ Verification problem where all candidates have the same scoring attribute values.

\begin{thm}\label{nphard}
    Fair Top-$k$ Verification is NP-hard for an arbitrary $n_p$.
\end{thm}

\begin{proof}
    For an arbitrary instance of the Set Cover problem, we are given a universe $\mathcal{U} = \{u_1, u_2, \ldots, u_{n_p}\}$ and a collection $\mathcal{S} = \{s_1, s_2, \ldots, s_n\}$ where each $s_\ell$ ($1 \leq \ell \leq n$) is a subset of $\mathcal{U}$. Our goal is to find the smallest subset of $\mathcal{S}$ whose union equals $\mathcal{U}$.

    Now, we construct an instance of Fair Top-$k$ Verification with the given instance of Set Cover. For each element $u_j$ in $U$, a group $\mathcal{G}_j$ is created that is also a protected group in the Fair Top-$k$ Verification problem. Then, for each subset $s_\ell$ in $S$ such that $s_\ell = \{u_x, u_y, u_z,\ldots \}$, we create a candidate $c$ with $A(c) = \{\mathcal{G}_x, \mathcal{G}_y, \mathcal{G}_z, \ldots \}$. These give us a set of $n$ candidates. The values of the scoring attributes are assigned to be the same across different candidates (e.g., $p(c) = (0, 0)$ for all $c$). Now we have an instance of the Fair Top-$k$ Verification problem in 2-D, with $n$ candidates and $n_p$ groups, and all groups are protected groups. To solve the given instance of the Set Cover problem, we set $L_{k}^{\scriptscriptstyle \mathcal{G}_j} = 1$ and $U_{k}^{\scriptscriptstyle \mathcal{G}_j} = k$ for all $1 \leq j \leq n_p$, and run the algorithm for the Fair Top-$k$ Verification problem with an arbitrary $w$ (e.g., $w = (0.5, 0.5)$) for all $1\leq k\leq n$. The smallest $k$ that successfully verifies the fairness of the weight vector is reported. Clearly, this reduction takes polynomial time.

    To show the correctness this reduction, note that for the given fairness constraints, if one successfully verifies the fairness of the given weight vector for a specific $k$ value, we will be able to derive a top-$k$ subset $\tau_k^w$ of the size $k$ such that each protected group, $\mathcal{G}_j$, must be covered by at least one $A(c)$ of $c\in \tau_k^w$. Moreover, consider the one-to-one correspondence between an element $u_j$ in the universe and a protected group $\mathcal{G}_j$, and between a subset $s_\ell$ of $S$ and a candidate $c$. A sub-collection of $S$ of size $k$ whose union covers all elements in $\mathcal{U}$ can be easily derived from $\tau_k^w$. In such cases, the smallest sub-collection of $S$ which covers $\mathcal{U}$ corresponds to a solution to the Fair Top-$k$ Selection problem with the smallest possible $k$.
    
    There is one subtlety that we need to discuss. An algorithm for the Fair Top-$k$ Verification problem can only confirm that a weight vector is fair if there exists a $\tau_k^w$ satisfying the fairness constraints. However, for this input $w$, it is not obvious how to derive a solution of the Set Cover problem as the $\tau_k^w$ satisfying the fairness constraints is still unknown. Fortunately, the decision version of the Set Cover problem, which only considers the size of the smallest sub-collection of $S$ that covers $U$, is known to be NP-hard, which is enough for our purpose.
\end{proof}

\begin{rmk}
    The previous results~\cite{cai2025finding} establish that the problem is solvable in polynomial time for any constant $d$ when $n_p = 1$. In contrast, our reduction shows that for an arbitrary $n_p$, the problem is NP-hard even for $d=2$. Consequently, a polynomial-time algorithm for general problem instances is unlikely to exist, even in constant dimensions.
\end{rmk}

\subsection{Hardness results for small $k$}
Our NP-hardness result considers the case where $k$ can be large. By the ``Small $k$ Opportunity'', it was shown in~\cite{cai2025finding} that if $k$ is sufficiently small, some lower bounds are breakable. In fact, if $k$ is a constant, there is a naive polynomial-time algorithm for the Fair Top-$k$ Verification problem. It runs in $O(n^k\cdot n_p)$ time by testing each possible combination of $k$ candidates, which works for the worst case where all candidates have the same scoring attribute values. Moreover, our NP-hardness result only applies to the case where the number of groups is large. In practice, however, the total number of groups is small in most datasets. In this section, we show that for a constant $k\geq 2$ and a moderate value of $n_p$ (i.e., $n_p = \alpha \log n$ for a constant $\alpha$), the problem may not admit a better (than the naive) algorithm achieving practically efficiency (i.e., an algorithm with a run time of $O(n\cdot \text{polylog}(n))$), based on hardness assumptions in the theory of fine-grained complexity \cite{williams2018some}.

To begin with, let us introduce two problems, Orthogonal Vectors (OV) and its generalization $t$-OV (often referred to as $k$-OV in the literature~\cite{williams2018some}), whose hardness assumptions play key roles in fine-grained complexity and which we will use in our reductions. 

The OV problem is defined as follows: Let $n_p = \alpha \log n$. Given two sets $A,B\subseteq \{0,1\}^{n_p}$ with $|A| = |B| = n$, the goal is to determine whether there exist $a \in A, b\in B$ so that $a \cdot b = 0$ where $a \cdot b = \sum_{i=1}^{n_p} a_i\cdot b_i$. For such a problem, we have the following conjecture~\cite{williams2005new, williams2018some}:

\begin{hyp}[OV Hypothesis]
    For every constant $\delta > 0$, there is a constant $\alpha > 0$ such that OV cannot be solved in $O(n^{2-\delta})$ time on instances with $n_p = \alpha \log n$.
\end{hyp}

An extension of the OV problem, $t$-OV problem, is defined as follows: Let $t$ be a constant and $n_p = \alpha \log n$. Given $t$ sets $A_1,A_2, \ldots, A_t\subseteq \{0,1\}^{n_p}$ with $|A_i| = n$ for all $i$, the goal is to determine whether there exist $a^1 \in A_1, a^2 \in A_2, \ldots a^t\in A_t$ so that $a^1 \cdot a^2 \cdot \ldots \cdot a^t = 0$ where $a^1 \cdot a^2 \cdot \ldots \cdot a^t = \sum_{i=1}^{n_p} \prod_{\ell=1}^t a_i^\ell$. For such a problem, we have the conjecture as follows:

\begin{hyp}[$t$-OV Hypothesis]
    For every constant $\delta > 0$, there is a constant $\alpha > 0$ such that $t$-OV cannot be solved in $O(n^{t-\delta})$ time on instances with $n_p = \alpha \log n$.
\end{hyp}

The OV and $t$-OV Hypotheses are popular due to their connections to the strong exponential time hypothesis (SETH) \cite{impagliazzo2001complexity}. These connections were established in \cite{williams2005new}.

With the two hypotheses introduced, we begin by establishing the following (conditional) lower bound based on the OV hypothesis.

\begin{thm}\label{thm:ov}
    For every constant $\delta > 0$, there is a constant $\alpha > 0$ such that Fair Top-$k$ Verification cannot be solved in $O(n^{2-\delta})$ time for $n_p = \alpha\log n$, assuming the OV hypothesis.
\end{thm}

\begin{proof}
    Given the two sets $A,B\subseteq \{0,1\}^{n_p}$, we first create a set of groups for our instance of the Fair Top-$k$ Verification problem, denoted as $\{\mathcal{G}_1,\mathcal{G}_2,\ldots,\mathcal{G}_{n_p}, \mathcal{G}_{n_p+1}, \mathcal{G}_{n_p+2}\}$. All groups are protected groups in our instance. For each element $a$ in $A$ such that $a = (a_1, a_2, \dots, a_{n_p})$, we create a candidate $c$ with $A(c) = \{a_1 \cdot \mathcal{G}_1, a_2 \cdot \mathcal{G}_2, \ldots, a_{n_p} \cdot \mathcal{G}_{n_p}, \mathcal{G}_{n_p + 1}\}$, where $a_i \cdot \mathcal{G}_i$ means $\mathcal{G}_i$ is an element of $A(c)$ if $a_i = 1$ and is not an element of $A(c)$ if $a_i = 0$. Similarly, for each element $b$ in $B$ such that $b = (b_1, b_2, \dots, b_{n_p})$, we create a candidate $c$ with $A(c) = \{b_1 \cdot \mathcal{G}_1, b_2 \cdot \mathcal{G}_2, \ldots, b_{n_p} \cdot \mathcal{G}_{n_p}, \mathcal{G}_{n_p + 2}\}$. Again, as was done in the proof of Theorem~\ref{nphard}, the values of scoring attributes are assigned to be the same across different candidates (e.g., $p(c) = (0, 0)$ for all $c$). Then, we set $L_{k}^{\scriptscriptstyle \mathcal{G}_j} = 0$ and $U_{k}^{\scriptscriptstyle \mathcal{G}_j} = 1$ for all $1 \leq j \leq n_p$, with $L_{k}^{\scriptscriptstyle \mathcal{G}_{n_p + 1}} = 1$ and $L_{k}^{\scriptscriptstyle \mathcal{G}_{n_p + 2}} = 1$ while $U_{k}^{\scriptscriptstyle \mathcal{G}_{n_p + 1}}\geq 1$ and $U_{k}^{\scriptscriptstyle \mathcal{G}_{n_p + 2}}\geq 1$ are set arbitrarily. Finally, we set $k=2$ and take an arbitrary input weight vector (e.g. $w = (0.5, 0.5)$). This gives us an instance of the Fair Top-$k$ Verification problem with $2n$ candidates and $n_p+2$ protected groups. Clearly, the reduction takes $O(n\log n)$ time with $n_p=O(\log n)$, as the time is bounded by $O(n\cdot n_p)$ for creating candidates.

    To show the correctness of this reduction, note that if one successfully verifies the fairness of the weight vector for such an instance of the Fair Top-$k$ Selection problem (with $k=2$), one must also be able to derive a set of two candidates, $c_1$ and $c_2$, such that $A(c_1)\cap A(c_2) = \emptyset$. Moreover,  $\mathcal{G}_{n_p+1}$ and $\mathcal{G}_{n_p+2}$ must be an element of either $A(c_1)$ or $A(c_2)$. Assume without the loss of the generality that $\mathcal{G}_{n_p+1} \in A(c_1)$ and $\mathcal{G}_{n_p+2} \in A(c_2)$. By our construction, this means that $c_1$ corresponds to an element in $A$ and $c_2$ corresponds to an element in $B$. Next, one can recover the corresponding element $a \in A$ by $A(c_1)$ by setting $a_i = 1$ if $\mathcal{G}_i \in A(c_1)$ and $a_i = 0$ otherwise. The element $b \in B$ can also be recovered in the same way. Given $A(c_1)\cap A(c_2) = \emptyset$, we now have $a \cdot b = 0$ as $a_i \cdot b_i = 0$ for all $1 \leq i \leq n_p$.

    Conversely, if the algorithm reports a failure, that means there does not exist $a \in A, b\in B$ such that $a \cdot b = 0$, as each candidate corresponds exactly one element in either $A$ or $B$ by our construction. This completes the proof.
\end{proof}

Moreover, using a similar approach to the one in Theorem~\ref{thm:ov}, we are able to show a stronger lower bound based on the $t$-OV hypothesis.

\begin{col}\label{col:kov}
    For every constant $\delta > 0$, there is a constant $\alpha > 0$ such that Fair Top-$k$ Verification cannot be solved in $O(n^{k-\delta})$ time for $n_p = \alpha\log n$, assuming the $t$-OV hypothesis.
\end{col}

\begin{proof}
    Given an arbitrary instance of the $t$-OV problem, we set $t=k$. For the $k$ sets $A_1,A_2, \ldots, A_k\subseteq \{0,1\}^{n_p}$, we first create a set of groups for our instance of the Fair Top-$k$ Selection problem, denoted as $\{\mathcal{G}_1,\mathcal{G}_2,\dots,\mathcal{G}_{n_p}, \mathcal{G}_{n_p + 1}, \mathcal{G}_{n_p + 2},\ldots, \mathcal{G}_{n_p + k} \}$, where all groups are protected. For each element $a$ in $A_\ell$ such that $a = (a_1, a_2, \dots, a_{n_p})$, we create a candidate $c$ with $A(c) = \{a_1 \cdot \mathcal{G}_1, a_2 \cdot \mathcal{G}_2, \ldots, a_{n_p} \cdot \mathcal{G}_{n_p}, \mathcal{G}_{n_p+\ell}\}$. Then, we set $L_{k}^{\scriptscriptstyle \mathcal{G}_j} = 0$ and $U_{k}^{\scriptscriptstyle \mathcal{G}_j} = k-1$ for all $1 \leq j \leq n_p$, as well as $L_{k}^{\scriptscriptstyle \mathcal{G}_{j+\ell}} = 1$ for all $1 \leq \ell \leq k$. Finally, similar to the proof of Theorem~\ref{thm:ov}, $U_{k}^{\scriptscriptstyle \mathcal{G}_{n_p + \ell}}\geq 1$ are arbitrarily set for all $1 \leq \ell \leq k$, values of scoring attributes are assigned to be the same across different candidates, and an arbitrary weight vector is used as an input. Clearly, the reduction takes $O(nk\cdot n_p) = O(n\log n)$ time, as $k = O(1)$ and $n_p = O(\log n)$. Note that in our reduction, there are $n_p+k$ protected groups rather than $n_p$, but since $k=O(1)$, the extra $k$ protected groups can be absorbed by choosing a slightly larger $\alpha$.

    To show the correctness of this reduction, note that if one successfully verifies the fairness of the weight vector for such an instance of the Fair Top-$k$ Verification problem, one must be able to derive a set of $k$ candidates, $C'$, such that for all $1\leq j\leq n_p$, $\mathcal{G}_j$ can only be an element of $A(c)$ of at most $k-1$ candidates in $C'$. Moreover, for each $1\leq \ell \leq k$, $\mathcal{G}_{n_p+\ell}$ can only be an element of $A(c)$ of exactly one candidate in $C'$, as the $t$-OV problem requires selecting exactly one binary vector from each binary vector set. By the same method applied in Theorem~\ref{thm:ov}, one can recover $k$ elements $a^1 \in A_1, a^2 \in A_2, \ldots a^k, \in A_k$ such that $a^1 \cdot a^2 \cdot \ldots \cdot a^k = 0$. This is because for all $1\leq i\leq n_p$, we have $\prod_{\ell=1}^k a_i^\ell = 0$ as there are at most $k-1$ $a_i^\ell = 1$ among all possible $1 \leq \ell \leq k$. Conversely, if the algorithm reports a failure, there must not exist such a set of $k$ binary vectors, given the one-to-one correspondence between a candidate and an element in the given $k$ sets $A_1,A_2, \ldots, A_k$.
\end{proof}

\begin{rmk}\label{rmk:kov}
    For a constant $k\geq 2$ and $n_p = \alpha \log n$, the naive algorithm takes $O(n^k\cdot \log n)$ time. Our reduction shows that in such cases, significant improvements over the naive approach are unlikely. Consequently, the Fair Top-$k$ Selection problem may not admit practically efficient algorithms (i.e., an algorithm with a run time of $O(n\cdot \text{polylog}(n))$) for such cases. However, notice that the reduction only works for $n_p$ being large enough such that the input sets contain distinct binary vectors. This observation, in fact, implies a more efficient algorithm when $n_p$ is sufficiently small.
\end{rmk}

\section{Small $k$ opportunity revisited}
We have shown that the Fair Top-$k$ Selection problem may not admit efficient algorithms, even for a constant $k \geq 2$ and moderate values of $n_p$ ($n_p = \Theta(\log n)$). However, Remark~\ref{rmk:kov} implies that when $n_p$ is sufficiently small, the lower bound may be breakable. In fact, in a previous work~\cite{cai2025finding}, it is established that when $n_p = 1$ and $k = O(\text{polylog}(n)))$, the problem is solvable in $O(n\cdot \text{polylog}(n))$ time for $d = 2,\,3$, confirming the ``Small $k$ Opportunity''. In this section, we further show that when $n_p$ and $k$ is sufficiently small, that is, $n_p = O(1)$ and $k = O(\text{polylog}(n)))$, the Fair Top-$k$ Verification problem can be solved in linear time for $d = O(1)$. This result breaks the previous NP-hardness and $\Omega(n^{k-\delta})$ lower bounds for these cases. By leveraging the key insights of the algorithm, we largely recover the results from the ``Small $k$ Opportunity'' for the Fair Top-$k$ Selection problem, with only a polylogarithmic overhead.

\subsection{The case of small $n_p$ (and $k$)}\label{subsec:small-np}
Before describing the algorithm in detail, we introduce its main idea. Observe that the hardness barrier for the Fair Top-$k$ Verification problem stems from ties in scores, which makes a naive enumeration algorithm near-optimal when $n_p = \Theta(\log n)$. To break this barrier, we leverage the key observation that when $n_p = O(1)$, many candidates become indistinguishable with respect to their \textit{protected groups membership profile}, which represents all protected groups to which a candidate belongs. These candidates are ``exchangeable'' for a top-$k$ subset that satisfies the fairness constraints, as swapping a selected candidate with an unselected one of the same protected groups membership profile does not change the protected group member counts of the top-$k$ subset. Thus, potentially many top-$k$ subsets are equivalent for the fairness test, so they can be bypassed for runtime improvements.

Now, let us proceed to the algorithm. First, an arbitrary top-$k$ subset is identified using a selection algorithm, with the score of each candidate under the given $w$ being computed. Let $\tau_k$ denotes this top-$k$ subset, and let $c_k \in \tau_k$ be the candidate with the lowest score in $\tau_k$, that is, the $k$th candidate under $w$ in the candidate set $C$. Next, two sets, $M_1\subseteq \tau_k$ and $M_2 \subseteq C \setminus \tau_k$, are identified such that every candidate $c$ in $M_1$ or $M_2$ has the same score as $c_k$ (i.e., $w \cdot p(c) = w \cdot p(c_k)$ for all $c \in M_1$ and $c \in M_2$). Notably, for every candidate in $\tau_k \setminus M_1$, their scores are strictly higher than that of $c_k$ so that they must appear in the top-$k$ subset regardless of the tie-breaking outcome. Therefore, we initialize the protected group member counts by $\tau_k \setminus M_1$.

Since there are $k - |M_1|$ candidates that are determined, our next step is to select $t = |M_1|$ candidates among $M_1 \cup M_2$ such that the fairness constraints are satisfied. Recall our observation that two candidates with the same protected groups membership profile are equivalent for a fair top-$k$ subset. With $n_p$ protected groups, there are $2^{n_{\scriptscriptstyle p}}$ distinct protected groups membership profiles. We encode each membership profile using a binary vector converted into an integer, denoted by the function $I(\cdot)$. For example, assuming $n_p = 5$, a candidate with membership profile $A(c) = \{\mathcal{G}_1,\mathcal{G}_3,\mathcal{G}_5 \}$ can be encoded into a binary vector $10101$, which is further converted as $I(A(c)) = 21$. Specially, if a candidate does not belong to any protected group, $I(A(c)) = 0$. The inverse of this operation, which coverts an integer to a protected groups member profile, is denoted as $I^{-1}(\cdot)$. In fact, a similar operation that converts a binary vector into a membership profile is used in the proof of Theorem~\ref{thm:ov}. Since candidates with the same protected groups membership profile are equivalent for the fairness test, only the number of candidates selected from each membership profile matters. Let $m_j$ denote the number of candidates with $I(A(c)) = j$, for a top-$k$ subset, one must have values of $n_j$ such that $\sum_{j=0}^{2^{n_p} - 1}m_j = t$. All valid value assignments to $m_j$ whose sum is $t$ can be enumerated using a backtracking algorithm.

Algorithm~\ref{alg:fairtopkver} provides the pseudocode for our algorithm, while Algorithm~\ref{alg:backtrack} provides the pseudocode for the backtrack subroutine. We use a dictionary $N$ to maintain the member counts for each protected group during the tree search. The array $S$ stores the total number of candidates for each protected groups membership profile, while $S'$ maintains the number of selected candidates for each profile during the search. Whenever the tree search reaches a leaf node, that is, when exactly $t$ candidates are selected, we use $N$ to test against the input fairness constraints $L=\{ L_{k}^{\scriptscriptstyle \mathcal{G}_j} \mid 1\leq j \leq n_p \}$ and $U=\{ U_{k}^{\scriptscriptstyle \mathcal{G}_j} \mid 1\leq j \leq n_p \}$. The tree search backtracks whenever the leaf node does not yield a top-$k$ subset satisfying the fairness constraint, or the number of candidates to be selected for a membership profile exceeds the limit specified in $S$.

\begin{algorithm}[tb]
    \caption{Algorithm for Fair Top-$k$ Verification}\label{alg:fairtopkver}
    \begin{algorithmic}[1]
        \Procedure{FairTopKVerification}{$C$, $k$, $n_p$, $L$, $U$, $w$}
        \State $\tau_k \gets$ \Call{Select}{$C$, $k$, $w$}
        \State $c_k\gets $ the $k$th candidate under $w$
        \State $M_1 \gets \{c \mid c\in \tau_k \text{ and } w\cdot p(c) = w\cdot p(c_k) \}$
        \State $M_2 \gets \{c \mid c\in C \setminus \tau_k \text{ and } w\cdot p(c) = w\cdot p(c_k) \}$
        \State Let $N$ be a dictionary storing protected group member counts, initialized to $0$
        \For {$c \in \tau_k \setminus M_1$}
            \For {$g \in A(c)$}
            \State {\bf if} $g$ is protected {\bf then} $N[g] \gets N[g] + 1$
        \EndFor
        \EndFor
        \State Let $S$ be an integer array of size $2^{n_p}$, initialized to $0$
        \For {$c \in M_1 \cup M_2$}
        \State $S[I(A(c))] \gets S[I(A(c))] + 1$
        \EndFor
        \State \Return \Call{BacktrackTieBreaking}{$|M_1|$, $n_p$, $L$, $U$, $N$, $S$}
        \EndProcedure
    \end{algorithmic}
\end{algorithm}

\begin{algorithm}[tb]
    \caption{Backtracking algorithm for tie-breaking}\label{alg:backtrack}
    \begin{algorithmic}[1]
        \Procedure{BacktrackTieBreaking}{$t$, $n_p$, $L$, $U$, $N$, $S$}
        \State Let $T$ be a stack to store pairs of integers, initialized to empty
        \State Let $S'$ be an array of size $2^{n_p}$, initialized to $-1$
        \State $T.$\Call{Push}{$\{t, 0 \}$}
        \While {$T$ is not empty}
        \State $\{r, i\} \gets T.$\Call{Top}{}()
        \If {$i \leq 2^{n_p} - 1$}
        \If {$S'[i] < S[i]$}
        \State $s'\gets $\Call{Max}{$S'[i]$, $0$}; $S'[i]\gets S'[i] + 1$; $r\gets r - (S'[i] - s')$
         \For {$g \in I^{-1}(i)$}
            \State $N[g] \gets N[g] + (S'[i] - s')$
        \EndFor
        \Else
        \State \Call{Revert}{$N$, $S'$}; $S'[i] \gets -1$; $T.$\Call{Pop}{}() \Comment{Restore to the previous state}
        \State {\bf continue}
        \EndIf
        \Else {} {\bf continue}
        \EndIf

        \If {$r = 0$}
        \If {\Call{IsFair}{$L$, $U$, $N$}}
        \State \Return True
        \Else 
        \State \Call{Revert}{$N$, $S'$}; $S'[i] \gets -1$; $T.$\Call{Pop}{}() \Comment{Restore to the previous state}
        \EndIf
        \Else {} $T.$\Call{Push}{$\{r, i + 1 \}$}
        \EndIf
        \EndWhile
        \State \Return False
        \EndProcedure
    \end{algorithmic}
\end{algorithm}

For the runtime analysis of the algorithm, it is easy to show the following: 
\begin{thm}\label{thm:ver}
    With $n_p = O(1)$ and a sufficiently small $k$ (e.g., $k = \text{polylog}(n)$), Fair Top-$k$ Verification can be solved in $O(n \cdot d)$ time.
\end{thm}

\begin{proof}
    Our algorithm starts by using a selection algorithm to obtain an arbitrary top-$k$ subset, which takes $O(n \cdot d)$ time. In the backtracking subroutine, each node of the search tree has at most $t + 1$ children, as there are at most $t$ candidates to be selected. The depth of the search tree is bounded by $2^{n_p}$, as each level corresponds to a unique protected groups membership profile. 
    
    At each node, updating $N$ takes $O(n_p)$ time, and if the node is a leaf, checking the fairness constraints also takes $O(n_p)$. Since $n_p = O(1)$, these are constant-time operations. Therefore, the total runtime of the backtracking subroutine is $O(t^{2^{n_p}})$. Given that $n_p = O(1)$, $2^{n_p}$ is also a constant. Since $t \leq k$, for a sufficiently small $k$ (e.g., $k = \text{polylog}(n)$), the term $O(t^{2^{n_p}})$ is $o(n\cdot d)$. (E.g., when $k = \text{polylog}(n)$, we have $O(t^{2^{n_p}}) = O(\text{polylog}(n)) = o(n\cdot d)$ since $t\leq k$ and $2^{n_p} = O(1)$.) As all other operations take $O(n)$ time (with $n_p = O(1)$), the initial selection step dominates the computational cost, resulting in an overall time complexity of $O(n \cdot d)$.
\end{proof}

\begin{rmk}\label{rmk:backtrack}
    Given that the runtime of the backtracking subroutine is $O(t^{2^{n_p}})$, our algorithm may remain efficient even for large $k$, provided that $t$ is sufficiently small. Furthermore, Algorithm~\ref{alg:backtrack} can be optimized by excluding protected groups membership profiles that contain zero candidates. However, since the main purpose of this algorithm is to demonstrate the key component for breaking the lower bound, we defer the discussion of this optimization and other heuristics for runtime improvements to Section~\ref{subsec:backtrack}.
\end{rmk}

\subsection{Augmenting the $k$-level-based algorithm}\label{subsec:theoretical-klevel}
Having developed an efficient algorithm for the Fair Top-$k$ Verification problem when both $k$ and $n_p$ are small, we now turn to the Fair Top-$k$ Selection problem. In the case of multiple protected groups, the backtracking subroutine is utilized to address the case of ties. Furthermore, we show how to augment the theoretically efficient $k$-level-based algorithm~\cite{cai2025finding} to minimize the $w$ difference and utility loss.

\paragraph{Multiple protected groups.}
In~\cite{cai2025finding}, the theoretical efficient $k$-level-based algorithm for $n_p = 1$ (without an optimization objective) works by traversing the cells of the $(k-1)$-level. During the traversal, it keeps track of the number of protected group members in the top-$k$ subset, and uses this count to test against the fairness constraint. With multiple protected groups, we adapt the method for the single protected group to keep track of the numbers of members for each protected group, and use them to against the fairness constraints whenever one of the numbers changes.

In the presence of ties, note that a top-$k$ subset with ties typically corresponds to a low-dimensional cell of the $(k-1)$-level, which is formed by the intersection of hyperplanes~\cite{halperin2017arrangements, cai2025finding}. To facilitate efficient calls to the backtracking subroutine for finding the top-$k$ subset satisfying fairness constraints in the presence of ties, we additionally track, for each protected group, the number of hyperplanes that lie strictly above the current $(k-1)$-level cell. Moreover, each $(k-1)$-level cell also stores auxiliary information about its intersecting hyperplanes, allowing the number of tied candidates for each protected groups membership profile to be computed efficiently.

\paragraph{$\boldsymbol{w}$ difference.}
For minimizing the $w$ difference, we focus on the cell which corresponds to a top-$k$ subset satisfying the fairness constraints. We refer to such a cell as a \textit{fair cell}. Note that if two weight vectors whose corresponding downward-directed rays in the dual space hit the same cell of the $(k-1)$-level, their top-$k$ subsets are the same~\cite{cai2025finding}. Consequently, to minimize the $w$ difference, one only has to find a weight vector that is closest to $w^o$ (under the $L_1$ distance metric) with its corresponding downward-directed ray hitting a fair cell. Equivalently, this problem becomes one of finding a weight vector closest to $w^o$ for each fair cell, and selecting an overall closest one among all fair cells. 

Let $\mathcal{F}$ denotes a fair cell that is projected onto the first $d - 1$ coordinates, and $\mathcal{H}$ be the set of boundary hyperplanes of $\mathcal{F}$. The fair cell $\mathcal{F}$ can be regarded as the intersection of the positive halfspaces induced by the hyperplanes in $\mathcal{H}$. To find a weight vector that is closest to $w^o$ under the $L_1$ metric, using the known technique that reformulates $L_1$-norm minimization as a linear program, we have that

\begin{equation}
\begin{aligned}
\text{min }  &\sum_{i=1}^{d} \phi_{i} &\\
\text{s.t. }     & \sum_{i=1}^{d - 1} h_iw_i + h_d\geq 0, &\forall h\in \mathcal{H}\\
                 & \phi_{i} \geq w_i - w^o_i, &i=1 ,\dots, d\\
                 & \phi_{i} \geq  w^o_i - w_i, &i=1 ,\dots, d\\
\end{aligned}
\end{equation}
To simplify the presentation, constraints on the weight vector (non-negativity and $||w||_1 = 1$) are omitted from the formulation. In this linear program, each $\phi_i$ can be regarded as the absolute difference between $w$ and $w^o$ at the coordinate $i$, and the objective function seeks to minimize the sum of the $\phi_i$. Notably, since $\mathcal{F}$ is a projection of a cell onto the first $d - 1$ coordinates, a hyperplane of $\mathcal{H}$ is of dimension $d-1$. For an $\mathcal{F}$ that only partially overlaps $V$, the weight vector space constraints implied by $V$ can also be easily incorporated into the linear program. See Figure~\ref{fig:theoretical-klevel} for an illustrative example. A fair cell of lower dimensions can also be handled similarly.

\begin{figure}[tbh!]
    \centering
    \resizebox{0.3\columnwidth}{!}{\begin{tikzpicture}
        \begin{axis}[y dir=reverse, xmin=0, xmax=1, ymin=0, ymax=1, zmin=0, zmax=1, 
            colormap = {slategraywhite}{rgb255=(220,220,220) rgb255=(180,180,180) rgb255=(140,140,140)}]
            \addplot3 [patch, patch type=triangle, draw=black, shader=flat, thick]
                coordinates { 
                    (0, 0, 0.8) (0.2, 0.0, 0.5) (0.0, 0.3, 0.65) 
                    (0.2, 0.0, 0.5) (0.0, 0.3, 0.65) (0.3, 0.3, 0.4)
                    (0.3, 0.3, 0.4) (0.2, 0.0, 0.5) (1.0, 0.0, 0.75)
                    (0.3, 0.3, 0.4) (0.0, 1.0, 0.85) (0.0, 0.3, 0.65)
                    (0.3, 0.3, 0.4) (0.6, 0.4, 0.42) (1.0, 0.0, 0.75)
                    (0.0, 1.0, 0.85) (0.4, 0.6, 0.42) (0.3, 0.3, 0.4)
                    (0.3, 0.3, 0.4)  (0.6, 0.4, 0.42) (0.4, 0.6, 0.42)
                };
            \addplot3[color=black, dashed, very thick]
                coordinates { (0.4, 0.15, 0.0) (0.1, 0.2, 0.0) (0.2, 0.6, 0.0)  (0.5, 0.4, 0.0)  (0.4, 0.15, 0.0)};
            \addplot3[color=black, dotted, very thick]
                coordinates { (0.3, 0.3, 0.0)  (0.6, 0.4, 0.0) (0.4, 0.6, 0.0)  (0.3, 0.3, 0.0) };

            \node [circle, fill=black, inner sep=1.5pt] at (axis cs:0.35, 0.25, 0.0) (wo){};
            
            \node at (axis cs:0.25, 0.35, 0.0) {\large $V$};
            \node at (axis cs:0.8, 0.7, 0.0) {\large $\mathcal{F}$};
            \node at (axis cs:0.6, 0.2, 0.0) {\large $w^o$};

            \node at (axis cs:0.8, 0.7, 0.0) (FA){};
            \node at (axis cs:0.5, 0.5, 0.0) (FB){};

            \node at (axis cs:0.57, 0.22, 0.0) (wA){};
            \node at (axis cs:0.35, 0.25, 0.0) (wB){};

            \draw [-{Latex}] (FA) -- (FB);
            \draw [-{Latex}] (wA) -- (wB);
        \end{axis}
    \end{tikzpicture}}
\caption{$(k-1)$-level (triangular mesh) and $V$ (dashed region on the $x$-$y$ plane) in 3-D. The dotted region (on the $x$-$y$ plane) represents a projected fair cell $\mathcal{F}$ and the point inside $V$ represents the input reference weight vector $w^o$. When minimizing the $w$ difference, the resulting fair weight vector will lie on the boundary of $\mathcal{F}$. In contrast, minimizing the utility loss enables a stable weight vector by selecting a weight vector inside $\mathcal{F}$ and away from its boundary.}\label{fig:theoretical-klevel}
\end{figure}
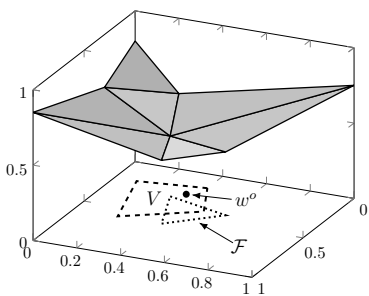

\paragraph{Utility loss.}
Just as in the case for minimizing the $w$ difference, only fair cells need to be considered for minimizing the utility loss. To do this, alongside the numbers of members for each protected group, we keep track of the utility of the top-$k$ subset, which is a summation of scores of candidates under $w^o$ within the top-$k$ subset. Since minimizing the utility loss is equivalent to maximizing the utility of the top-$k$ subset, the maximum utility of top-$k$ subset satisfying the fairness constraints is maintained during the traversal.

Special attention is needed for handling top-$k$ subsets in the presence of ties, as different tie-breaking outcomes may result in top-$k$ subsets with different utility (recall Example~\ref{eg:utility}). Since candidates of the same protected groups membership profile are exchangeable, a greedy strategy can be employed: once the backtracking algorithm identifies a fair tie-breaking outcome, for each protected groups membership profile, candidates with the highest scores are selected to achieve the maximum utility. To do so efficiently, for each protected groups membership profile, we sort the candidate by its score under $w^o$ with a non-increasing order. In fact, only partial sorting is needed since at most $t \leq k$ candidates can be selected for each protected groups membership profile. Then, an array of prefix sums over these sorted scores is computed, which allows fast retrieval of the maximum utility when selecting any given number of candidates from that protected groups membership profile. The utility of selected candidates is maintained using the same method for maintaining the protected group membership counts during the backtracking search.

Finally, one may observe that a weight vector minimizing the utility loss may not be unique, as the utility loss only considers the top-$k$ subset induced by the weight vector, and two weight vectors whose corresponding downward-directed rays in the dual space hit the same cell of the $(k-1)$-level have the same top-$k$ subsets. Of course, one can decide a single weight vector by additionally minimizing the $w$ difference, that is, finding the closest fair weight vector with the minimum utility loss. However, the utility loss optimization objective enables us to consider the other important aspect of the top-$k$ selection, that is, the stability of a scoring function \cite{asudeh2018obtaining}. In many applications, one may wish to find a stable scoring function, that is, small perturbations of the weights do not change the selected top-$k$ subset. Unfortunately, a fair weight vector that minimizes the $w$ difference is not stable by its nature, as a small perturbation towards $w^o$ would yield top-$k$ subsets violating the fairness constraints. With the utility loss objective, one is able to obtain a stable fair weight vector by solving the following linear program:
\begin{equation}\label{eq:lp-utility}
\begin{aligned}
\text{max }  &\xi &\\
\text{s.t. } & \sum_{i=1}^{d - 1} h_iw_i  + h_d \geq \xi, & \forall h\in \mathcal{H}
\end{aligned}
\end{equation}
Again, $\mathcal{H}$ is the set of boundary hyperplanes of a (projected) fair cell $\mathcal{F}$. This linear program enforces a margin of $\xi$ between the weight vector and the cell boundary, ensuring that the corresponding downward-directed ray hits the same cell for a weight vector under small perturbations bounded by $\xi$. As a result, the resulting fair weight vector is more stable with respect to perturbations (see Figure~\ref{fig:theoretical-klevel}).

\paragraph{Runtime analysis.}
The augmentations above consider the presence of multiple protected groups and the minimization of the $w$ difference or the utility loss. Although the former increases the number of variables maintained during the $k$-level cell traversal and uses a backtracking algorithm for tie-breaking, while the latter introduces additional overhead for fair cells, the overall run time does not increase by a lot when both $k$ and $n_p$ are small. Formally, we have the following:

\begin{thm}
    With $n_p = O(1)$, $k = O(\textnormal{polylog}(n))$ and $l = |V| = O(\textnormal{\text{polylog}}(n))$, Fair Top-$k$ Selection can be solved in $\tilde{O}(n)$ time for $d=2,3$ ($\tilde{O}(\cdot)$ hides a polylog $n$ factor). For constant dimensions $d\geq 4$, it can be solved in $\tilde{O}(n^{\lfloor d/2 \rfloor})$ (expected) time.
\end{thm}

\begin{proof}
    It is known that for $n_p = 1$ and $k = O(\textnormal{polylog}(n))$, Fair Top-$k$ Selection can be solved in $\tilde{O}(n)$ time for $d=2,3$, and in $\tilde{O}(n^{\lfloor d/2 \rfloor})$ (expected) time for a constant $d\geq 4$~\cite{cai2025finding}. Consider the additional overhead introduced by the augmentations. For the multiple protected groups, the additional time spent for maintaining additional variable is $O(n_p)$, which are constant-time operations as $n_p = O(1)$. For the backtracking subroutine, the additional overhead is $O(j + k^{2^{n_p}})$, where $j$ is the number of intersecting hyperplanes at a cell. Since each additional intersecting hyperplane (making the configuration divert from a general position one) decreases the number of cells, the first part of the cost only adds $O(1)$ per cell in the worst case. Since $k = O(\textnormal{polylog}(n))$ and $n_p = O(1)$, the second part of the cost only adds a polylog $n$ factor of run time to each cell. For minimizing $w$ difference and the utility loss, an additional linear program has to be solved. Since the number of variables is at most $2d$, this can be achieved in $O(|\mathcal{H}| + l)$ time using a linear-time algorithm for fixed-dimensional linear programming \cite{megiddo1984linear,seidel1991small}. Since each cell can be subdivided into simplices~\cite{henk2017basic}, each with $O(1)$ boundary hyperplanes in fixed dimensions, and the number of cells remains bounded by the worst-case number after the subdivision, the first part of the cost only adds $O(1)$ per cell in the worst case. The additional cost of solving the linear program for each cell is then $O(\textnormal{polylog}(n))$, as $l = O(\textnormal{\text{polylog}}(n))$. Minimizing the utility loss also introduces additional overheads to the backtracking subroutine. However, the overhead per cell is only $O(\log k) = O(\log \log n)$ due to the partial sorting.
\end{proof}

\section{Practical algorithms}
As noted in Remark~\ref{rmk:backtrack}, the backtracking algorithm presented in Section~\ref{subsec:small-np} primarily illustrates the key component for breaking the lower bounds. Since the algorithm is used as a subroutine in our practical algorithms for the Fair Top-$k$ Selection problem, in this section, we present an efficient implementation of it, incorporating several practical heuristics for runtime improvements. As for solving the Fair Top-$k$ Selection problem, a previous work~\cite{cai2025finding} proposes an efficient two-pronged solution for $n_p = 1$ without an optimization objective: a $k$-level-based algorithm for small $k$ and a MILP-based algorithm for large $k$. Notably, instead of using the theoretically efficient $k$-level-based algorithm, practical variants that address various engineering challenges are employed. In this section, we show how to augment practical $k$-level-based algorithms to handle multiple protected groups and optimization objectives ($w$ difference and utility loss), and we discuss various engineering decisions and trade-offs in the augmentations. Finally, we show how to augment the mixed-integer linear programming-based (MILP-based) algorithm, which handles the case of large $k$ in the two-pronged solution, to accommodate multiple protected groups and optimization objectives.

\subsection{Engineering the backtracking algorithm for tie-breaking}\label{subsec:backtrack}
Recall in Section~\ref{subsec:small-np}, the backtracking algorithm regards each distinct protected groups membership profile as a level in the search tree. However, it is possible that some protected group membership profiles may have no candidates associated with them, and thus can be ignored. Let $\beta \leq 2^{n_p}$ be the number of distinct protected groups membership profiles with at least one candidate. The backtracking search then becomes one of enumerating all possible value assignments to $\beta$ non-negative integer variables whose sum is exactly $t$. By applying the ``stars and bars'' technique, one can show that the total number of such assignments is $\binom{t + \beta - 1}{\beta - 1}$, which is $O(t^{\beta - 1})$ when $\beta = O(1)$. Notably, since adding ``artificial'' sensitive attribute values corresponding to the intersections of multiple protected groups does not increase $\beta$, the algorithm incurs only a small overhead for handling intersectional fairness constraints.

To accelerate the backtracking, we also prune the search tree by considering the total number of remaining candidates during the search. At each level of the search tree, one can check the total number candidates available in all subsequent levels. The number of candidates selected at this level must be large enough such that the total number of selected candidates can reach $t$ if all remaining candidates in subsequent levels are selected. This information, that the total number candidates in the remaining levels, can be easily retrieved from an array of prefix sum over the numbers of candidates of each protected groups membership profile.

Furthermore, notice that when $n_p = 1$, the fairness test can be done with a greedy approach, which computes the tight lower and upper bounds on the number of protected group candidates for all possible top-$k$ subsets~\cite{cai2025finding}. These bounds are used to test against the fairness constraint. This approach is also employed in our implementation to deal the case of $n_p = 1$. When $n_p > 1$, this greedy approach is no longer sufficient to guarantee correctness. However, it can still be employed as a fast pruning step before invoking the backtracking search. For each protected group, we use the greedy approach to compute the corresponding lower and upper bounds on the number of candidates. If the interval defined by these bounds does not overlap the given fairness constraint for that group, one can safely conclude all top-$k$ subsets violate the fairness constraints.

Finally, to minimize the utility loss, an approach that tracks the utility of the top-$k$ subset during the backtracking search was proposed in Section~\ref{subsec:theoretical-klevel}. In our implementation, we take a simpler approach that, whenever the search reaches a leaf corresponding to a fair top-$k$ subset, its utility is computed from scratch using the prefix sum array of candidate scores for each protected groups membership profile, where the number of candidates of each membership profile is stored in $S'$ during the backtracking search (recall Algorithm~\ref{alg:backtrack}). Even though this approach increases the runtime overhead from $O(n_p)$ to $O(\beta)$, where $\beta$ denotes the number of distinct protected groups membership profiles with at least one candidate, the value of $\beta$ is also small in practice, and utility computation is only executed at fair leaves. Moreover, this approach is numerically more robust, as dynamically updating the utility during the search is more susceptible to round-off errors than recomputing it from scratch. Notably, when $n_p = 1$, a greedy approach, which selects a candidate with the highest possible score while complying with the fairness constraint, is sufficient for the correctness. This method is implemented for handling the $n_p = 1$ case.

\subsection{Practical $k$-level-based algorithms}\label{subsec:pract-klevel}
For $n_p = 1$ without an optimization objective, two practical $k$-level-based algorithms for small $k$ have already been given in~\cite{cai2025finding}: one for $d=2$ and one for $d\geq 3$. We first describe how to augment the 2-D $k$-level-based algorithm, and then move to the augmentations of the multi-dimensional $k$-level-based algorithm.

\subsubsection{2-D $k$-level-based algorithm}
For $n_p = 1$ without an optimization objective, the 2-D $k$-level-based algorithm~\cite{cai2025finding} is an adaptation of an algorithm in \cite{chan1999remarks} (briefly outlined in \cite{basch1996reporting}). The algorithm is a sweep-line algorithm using two kinetic tournament trees \cite{basch1996reporting}, $S_1$ and $S_2$, to keep track of the $k$-level cell. Viewing the $x$-coordinate of the sweep-line as ``time'', $S_1$ (resp. $S_2$) contains the $k$ (resp. $n-k$) lines lying above (resp. below) at the current time instant. The algorithm proceeds by processing internal queue events or updating queues during the sweep. During the sweep, the number of protected group member is maintained for the fairness test. And the tie-breaking is resolved with the following techniques: (1) Symbolic perturbation~\cite{edelsbrunner1990simulation} for event processing of kinetic tournament trees. (2) Carefully identifying pairs of intersecting lines to exchange between $S_1$ and $S_2$. (3) Guided tree traversal to collect all intersecting lines. (4) Greedy approach to compute tight lower and upper bounds on the number of protected group candidates for all possible top-$k$ subsets.

\paragraph{Multiple protected groups.}
Just as was done in the theoretical $k$-level-based algorithm, to accommodate multiple protected groups, we keep track of the member counts for each protected group within $S_1$ during the sweep. For the tie-breaking with $n_p > 1$, after collecting the intersecting lines found by the guided tree traversal of $S_1$ and $S_2$, the dictionary $N$, which maintains the member counts for each protected group during the backtrack search, is initialized using the protected group member counts within $S_1$ and the intersecting lines from $S_1$. The array $S$ can also be easily initialized by the set of intersecting lines (from $S_1$ and $S_2$). These steps together enable an invocation of the backtracking subroutine.

\paragraph{$\boldsymbol{w}$ difference.}
Since $d=2$ and $||w||_1 = 1$, for a given reference weight vector $w^o = (w_x^o, 1 - w_x^o)$ and any weight vector $w = (w_x, 1 - w_x)$, the $w$ difference can be written as $g^o(w) = ||w - w^o||_1 = 2|w_x - w_x^o|$. Consequently, finding a fair weight vector minimizing the  $w$ difference is equivalent to finding a fair weight vector minimizing the absolute difference in the first coordinate. In the context of the sweep-line algorithm, regarding $w^o$ as a vertical line $x = w_x^o$, the goal becomes finding the closest vertical line that yields a top-$k$ subset satisfying the fairness constraints. Based on this observation, one can employ a bidirectional sweep-line algorithm to find the fair weight vector minimizing the $w$ difference. Starting from a weight vector that is to the right (resp. left) of $w^o$ and is in $V$ (including $w^o$ itself), the sweep-line moves from left to right (resp. from right to left) and stops at the first position that yields a fair top-$k$ subset. The fair weight vector that minimizes the $w$ difference is obtained from the two resulting fair weight vectors produced by the bidirectional line-sweeping (see Figure~\ref{fig:opt1}).

This approach is efficient as it only (implicitly) constructs the part of $(k-1)$-level that is relevant to the search. However, it requires a kinetic tournament tree to be constructed twice, especially given that different sweeping directions affect the outcomes of symbolic perturbations. In the experiments for $n_p = 1$~\cite{cai2025finding} (as well as those in Section~\ref{subsubsec:exp2d-multi}), it was observed that the time spent on the line-sweeping process only accounted for only a small fraction of the total run time, making the construction of a kinetic tournament tree relatively expensive in practice. In our implementation, we take a simpler and more efficient approach which sweeps along one direction, starting from the left end-point defined by $V$ while tracking the minimum $w$ difference among the fair weights vectors. The previous observation is used as an early termination criterion, that if the sweep-line passes $w^o$ and encounters a fair weight vector, the algorithm terminates immediately.

\begin{figure}[tbh!]
    \captionsetup[subfigure]{justification=centering}
    \centering
\begin{subfigure}[b]{0.3\columnwidth}
    \centering
    \includegraphics[width=0.7\linewidth]{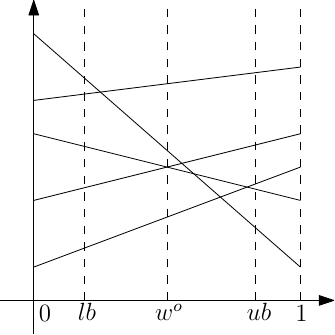}
    \caption{}\label{fig:opt1}
\end{subfigure}
\begin{subfigure}[b]{0.3\columnwidth}
    \centering
    \includegraphics[width=0.7\linewidth]{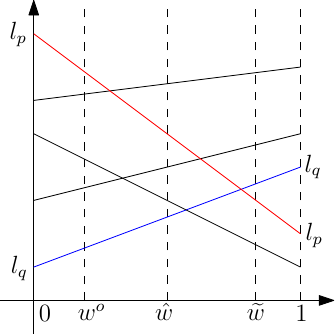}
    \caption{}\label{fig:opt2}
\end{subfigure}
\caption{\textbf{(a)} $w$ difference minimization in 2-D, where the interval represents $V$ and the vertical dashed line $w^o$ represents the reference weight vector. The bidirectional sweep-line algorithm works by sweeping from $w^o$ to $ub$ and to $lb$. \textbf{(b)} For the weight vectors $\hat{w}$ and $\widetilde{w}$ with $k = 3$, $l_p$ (resp. $l_q$) is in the top-$k$ subset of $\hat{w}$ (resp. $\widetilde{w}$) but not in that of $\widetilde{w}$ (resp. $\hat{w}$). Intuitively, $l_p$ lies above $l_q$ at $x = w_x^o$ since the two lines intersect at a point between $x = \hat{w}_x$ and $x = \widetilde{w}_x$, where their orders change.}
\end{figure}

\paragraph{Utility loss.}
At first glance, it appears that for finding a fair weight minimizing the utility loss, the entire $(k-1)$-level within $V$ must be swept. However, the bidirectional sweep-line algorithm can also be applied here, due to the following observation.

\begin{lem}\label{lem:2d-utility}
    In 2-D, consider a given a reference weight vector $w^o = (w_x^o, 1 - w_x^o)$ and two weight vectors $\hat{w} = (\hat{w}_x, 1 - \hat{w}_x)$ and $\widetilde{w} = (\widetilde{w}_x, 1 - \widetilde{w}_x)$, where $w_x^o < \hat{w}_x < \widetilde{w}_x$ or $w_x^o > \hat{w}_x > \widetilde{w}_x$. The utilities of any pair of their corresponding top-$k$ subsets, $\tau_k^{\hat{w}}$ and $\tau_k^{\widetilde{w}}$, must satisfy $U^o(\tau_k^{\hat{w}}) \geq U^o(\tau_k^{\widetilde{w}})$.
\end{lem}

\begin{proof}
    The proof is given only for the $w_x^o < \hat{w}_x < \widetilde{w}_x$ case, as the $w_x^o > \hat{w}_x > \widetilde{w}_x$ case is symmetric. Let $\tau_k^{\hat{w}}$ and $\tau_k^{\widetilde{w}}$ be an arbitrary top-$k$ subset of $\hat{w}$ and $\widetilde{w}$, respectively. If $\tau_k^{\hat{w}} \neq \tau_k^{\widetilde{w}}$ (otherwise, $U^o(\tau_k^{\hat{w}}) \geq U^o(\tau_k^{\widetilde{w}})$ trivially holds), one can find a pair of points such that each top-$k$ subset contains a distinct point from the pair. More formally, we have a pair of points, $p$ and $q$, such that $p\in \tau_k^{\hat{w}}$ and $q\in \tau_k^{\widetilde{w}}$ but $p\notin \tau_k^{\widetilde{w}}$ and $q\notin \tau_k^{\hat{w}}$. By the dual transformation, we have two lines $l_p:  y = (p_x - p_y)x + p_y$ and  $l_q: y = (q_x - q_y)x + q_y$, in which the value of the $y$-coordinate is a score under a weight vector. By $p\in \tau_k^{\hat{w}}$ and $q\notin \tau_k^{\hat{w}}$, we have $(p_x - p_y)\hat{w}_x + p_y \geq (q_x - q_y)\hat{w}_x + q_y$, since $p$ is among top-$k$ under $\hat{w}$ and $q$ is not. Similarly, By $q\in \tau_k^{\widetilde{w}}$ and $p\notin \tau_k^{\widetilde{w}}$, we have $(p_x - p_y)\widetilde{w}_x + p_y \leq (q_x - q_y)\widetilde{w}_x + q_y$. With $\hat{w}_x < \widetilde{w}_x$, we have that $(p_x - p_y)\Delta x \geq (q_x - q_y)\Delta x$ for any $\Delta x \leq 0$, since the direction of the inequality changes as the value of $x$ decreases from $\widetilde{w}_x$ to $\hat{w}_x$. Then, given $w_x^o < \hat{w}_x$, we have that $(p_x - p_y)(w_x^o - \hat{w}_x) \geq (q_x - q_y)(w_x^o - \hat{w}_x)$, and thus $(p_x - p_y)w_x^o + p_y \geq (q_x - q_y)w_x^o + q_y$. Because for any such pair of $p$ and $q$, the score of $p$ is at least as large as that of $q$ under $w^o$, the utility of $\tau_k^{\hat{w}}$, which is the sum of all candidate scores under $w^o$, is at least as large as that of $\tau_k^{\widetilde{w}}$ (see Figure~\ref{fig:opt2}).
\end{proof}

In other words, Lemma~\ref{lem:2d-utility} tells us that if two weight vectors lie on the same side of $w^o$, the one that is closer to $w^o$ along the $x$-coordinate has a utility loss that is lower than or equal to that of the other one, as the utility of its corresponding top-$k$ subset is at least as large. Based on this observation, the bidirectional line sweeping process can also stop at the first position that yields a fair top-$k$ subset, since further sweeping does not yield a top-$k$ subset with a higher utility. Of course, for practical efficiency, we use a unidirectional sweeping approach by tracking the highest utility of top-$k$ subsets during the sweep in our implementation. Lemma~\ref{lem:2d-utility} can still be leveraged as an early termination criterion, analogous to the case of $w$ difference minimization.

To obtain a stable fair weight vector, with the bidirectional sweep-line algorithm, one can continue the sweeping to identify the other end-point of the $(k-1)$-level cell (i.e. line segment), and then obtain the weight vector via the midpoint of the cell. Here, we introduce an alternative approach that, while potentially being less efficient, is easy to implement and more flexible. After obtaining the fair weight vector $w^f$ that minimizes the utility loss, we run a variant of our algorithm for the Fair Top-$k$ Verification problem (Algorithm~\ref{alg:fairtopkver}) to obtain its corresponding top-$k$ subset, $\tau_k^{\scriptscriptstyle f}$, that minimizes the utility loss. Then, the $k$th candidate, $c_k$, is obtained from $\tau_k^{\scriptscriptstyle f}$. (The tie-breaking involved here will be discussed later.) For any candidate $c \in \tau_k^{\scriptscriptstyle f} \setminus \{c_k\}$, its corresponding line in the dual space, $l_c$, must lie above $l_{c_{\scriptscriptstyle k}}$ associated with $c_k$. Consider the intersection point of $l_c$ and $l_{c_{\scriptscriptstyle k}}$, if the intersection point falls within the interval defined by $V$. It partitions $V$ into two intervals, and only one of them is valid as the corresponding downward-directed ray of the stable fair weight vector intersects $l_c$ first. Similar, for any candidate $c \in C \setminus \tau_k^{\scriptscriptstyle f}$, its corresponding line in the dual space must lie below the corresponding line of $c_k$, and a valid interval can be derived in the same manner. The intersection of these intervals define a weight vector subspace corresponding to the top-$k$ subset $\tau_k^{\scriptscriptstyle f}$, and we take the midpoint of the interval as the stable fair weight vector.

Apart from the time spent on backtracking search (which the sweep-line algorithm would also incur), this approach takes $O(n\cdot t)$ time, as there are $t$ candidates that are possibly the $k$th one due to ties in scores (recall Section~\ref{subsec:small-np}). In the worst case, $t$ can be as large as $k$, though this rarely happens in practice. Moreover, the approach can be applied as a post-processing step to any algorithm that finds a fair weight vector minimizing the utility loss, making it both flexible and decoupled from the specific Fair Top-$k$ Selection algorithm. Finally, it is easy to extend this method to higher dimensions, as will be shown shortly.

\subsubsection{Multi-dimensional $k$-level-based algorithm}\label{subsubsec:pract-klevel-md}
For $n_p = 1$ without an optimization objective, the multi-dimensional $k$-level-based algorithm~\cite{cai2025finding} adapts the algorithm in \cite{andrzejak1999optimization}. Starting from a cell and its top-$k$ subset, the algorithm explores adjacent cells by testing potential top-$k$ subsets, formed by swapping one top-$k$ element with an outsider. The test can be done by finding a separating hyperplane via a linear program. This enables enumeration of all valid top-$k$ subsets in a breadth-first manner. Via a lockless (i.e., without explicit locking mechanisms) implementation proposed in~\cite{cai2025finding} that parallelizes the testings of potential top-$k$ subsets, the algorithm can efficiently utilize a multi-core shared-memory system.

\paragraph{Multiple protected groups.}
Augmenting the algorithm to handle multiple protected groups is straightforward. Since the algorithm actually enumerates all valid top-$k$ subsets, one can directly count the number of members of each protected group within a top-$k$ subset, and test the numbers against the fairness constraints. This is a direct extension to the method used for the $n_p = 1$ case.

Additionally, we further improve the performance of the practical multi-dimensional $k$-level-based algorithm by incorporating a technique from~\cite{cai2025finding}, which was used to improve the baseline algorithm therein. When considering a swap between a current top-$k$ element and an outsider candidate, we examine the intersection of their corresponding hyperplanes in the dual space. If this intersection lies outside the region $V$ (after projecting the intersecting onto the first $d-1$ coordinates), it implies that the element being swapped out always has a higher score than the element being swapped in for all weight vectors within $V$. In this case, the swapping can be skipped. Since $l$ (the number of linear inequalities of $V$) is typically small, extreme points of $V$ are precomputed to speed-up intersection tests. This preprocessing incurs small overhead, especially when the dimensionality $d$ is low.

\paragraph{$\boldsymbol{w}$ difference.}
Just as was done in the theoretical $k$-level-based algorithm, one can minimize the $w$ difference by solving linear programs for fair cells. However, in our practical multi-dimensional $k$-level-based, the $(k-1)$-level cell is not explicitly determined, thus its set of boundary hyperplanes is not known. This issue can be addressed by utilizing the knowledge of the top-$k$ subset and adding a cut-off variable, which gives us the following linear program:
\begin{equation}
\begin{aligned}
\text{min }  &\sum_{i=1}^{d} \phi_{i} &\\
\text{s.t. } & p(c) \cdot w \geq \lambda, & \forall c\in \tau_k^{w}\\
             & p(c) \cdot w \leq \lambda, & \forall c\in C \setminus \tau_k^{w}\\
             & \phi_{i} \geq w_i - w^o_i, &i=1 ,\dots, d\\
             & \phi_{i} \geq  w^o_i - w_i, &i=1 ,\dots, d
\end{aligned}
\end{equation}
Notice that in this linear program, the top-$k$ subset, $\tau_k^{w}$, is fixed and determined by the previous step. Our goal is to find a weight vector that minimizes the $w$ difference within the cell. Thus, the top-$k$ subset remains unchanged throughout the optimization process. There are $2d + 1$ variables in this linear program. While the equation $\sum_{i=1}^d w_i = 1$ can be utilized to eliminate one dimension, the linear program still has $2d$ variables. Due to the factor of two in the number of variables, Seidel's LP algorithm \cite{seidel1991small} that we used to solve linear programs for a small number of variables may quickly become inefficient as the dimensionality increases. The run time of Seidel's LP algorithm includes a multiplicative term that grows as the factorial of the number of variables, so the algorithm is efficient only when the number of variables is a small constant. Moreover, our experimental results (Section~\ref{subsubsec:expmd-multi}) suggested that the number of fair cells is typically small, thus these linear programs may only account for only a small fraction of the total run time. In our implementation, we use the simplex algorithm to solve linear programs for minimizing the $w$ difference, as the algorithm is practically efficient and its performance is less sensitive to the number of variables.

Since finding a fair weight vector that minimizes the $w$ difference requires accounting for all fair cells, we also adapt the previous lockless parallel implementation~\cite{cai2025finding} to support this process. Specifically, each thread maintains a local optimal solution while traversing the $(k-1)$-level cells. Upon completion of the traversal, a final reduction step is used to determine the global minimum among these local optima. While a parallel reduction is possible~\cite{frigo2009reducers}, we employ a simple serial reduction in our implementation, as it is sufficiently efficient for typical multi-core shared-memory systems (up to a few hundred cores).

\paragraph{Utility loss.}
Comparing with $w$ difference minimization, minimizing the utility loss in our practical multi-dimensional $k$-level-based algorithm is relatively simple. For a top-$k$ subset that satisfies the fairness constraints, its utility can be directly computed, and the one with the highest utility can be maintained during the traversal. The lockless parallel implementation can also be adapted using the same method for $w$ difference minimization.

To find a stable fair weight vector, we extend the approach used in the 2-D $k$-level-based algorithm. Again, for a fair weight vector $w^f$ that minimizes the utility loss, its corresponding top-$k$ subset $\tau_k^{\scriptscriptstyle f}$ is obtained, and the $k$th candidate, $c_k$, is obtained from $\tau_k^{\scriptscriptstyle f}$. For any candidate $c \in \tau_k^{\scriptscriptstyle f} \setminus \{c_k\}$, the intersection of the hyperplane corresponding to  $c$ with that of $c_k$ can be determined by $(p(c) - p(c_k))\cdot w = 0$, as a hyperplane intersection corresponds to a tie in scores. Projecting this intersection onto the first $d - 1$ coordinates partitions the weight space into two subspaces, and only one of them is valid as the downward-directed ray of the stable fair weight vector intersects the hyperplane associated with $c$ first. For any candidate $c \in C \setminus \tau_k^{w^{\scriptscriptstyle f}}$, a valid subspace can be derived in the same manner. The set of bounding hyperplanes of these subspaces is actually a superset of $\mathcal{H}$, which is the set of boundary hyperplanes of the (protected) fair cell $\mathcal{F}$. Moreover, the intersection of these subspaces is $\mathcal{F}$ itself. Thus, the linear program~(\ref{eq:lp-utility}) can also be applied here, by replacing $\mathcal{H}$ with the hyperplanes determined in the previous steps. Since the linear program contains $d + 1$ variables (which can be reduced to $d$), Seidel's LP algorithm is employed. Solving the linear programs takes $O(t\cdot (l + n))$ time, with a linear time algorithm for linear programming in fixed dimensions being employed.

\subsection{Mixed-integer linear programming-based algorithm}
For $n_p = 1$ without an optimization objective, a mixed-integer linear programming-based (MILP-based) algorithm~\cite{cai2025finding} is used for large $k$, which is practically efficient despite being theoretically suboptimal. The algorithm uses a binary indicator variable, $\delta_c$, to encode whether a candidate $c$ is within a top-$k$ subset or not. Assume without loss of generality that each scoring attribute value is in the range of $[0, 1]$ (as discussed in~\cite{cai2025finding}), this variable satisfies the following inequalities:
\begin{equation}
    -1 \leq w \cdot p(c) - \lambda - \delta_c \leq 0,
\end{equation}
where $\lambda\in [0, 1]$ is a cut-off value such that the pair $(w,\lambda)$ can be regarded as a separating hyperplane. Then, the problem is formalized into a MILP to find a weight vector whose top-$k$ subset satisfying the fairness constraint, using the indicator variables to enforce it. The resulting MILP is solved using a state-of-the-art MILP solver~\cite{gurobi}.

\paragraph{Multiple protected groups.}
To extend the approach to the case where $n_p > 1$, we apply the method for encoding the fairness constraint for one protected group to each protected group. The resulting fairness constraints are formulated as
\begin{equation}
    L_{k}^{\scriptscriptstyle \mathcal{G}_j} \leq \sum_{c \in \mathcal{G}_j} \delta_c \leq U_{k}^{\scriptscriptstyle \mathcal{G}_j},
\end{equation}
for all $1\leq j \leq n_p$. All other constraints are kept as-is.

\paragraph{$\boldsymbol{w}$ difference.}
In Section~\ref{subsec:theoretical-klevel}, minimizing the $w$ difference is formalized as a linear program, making it easy to incorporate it into the mixed-integer linear program. By using the same formulation method, we have the following MILP:
\begin{equation}
\begin{aligned}
\text{min }  &\sum_{i=1}^{d} \phi_{i} &\\
\text{s.t. } &\sum_{c \in C} \delta_c = k &\\
             &\sum_{c \in \mathcal{G}_j} \delta_c \geq L_{k}^{\scriptscriptstyle \mathcal{G}_j}, & j=1,\dots, n_p\\
             &\sum_{c \in \mathcal{G}_j} \delta_c \leq U_{k}^{\scriptscriptstyle \mathcal{G}_j}, & j=1,\dots, n_p\\
             & w \cdot p(c) - \lambda - \delta_c \leq 0, & \forall c\in C\\
             & w \cdot p(c) - \lambda - \delta_c \geq -1, & \forall c\in C\\
             & \phi_{i} \geq w_i - w^o_i, &i=1 ,\dots, d\\
             & \phi_{i} \geq  w^o_i - w_i, &i=1 ,\dots, d\\
\end{aligned}
\end{equation}
with constraints on the weight vector $w$, the cut-off variable $\lambda$ (i.e., $\lambda\in [0, 1]$), and indicator variables $\delta_c$ (i.e., $\delta_c \in \{0, 1\}$ for all $c \in C$) also being incorporated accordingly (omitted from the above formulation to simplify the presentation).

\paragraph{Utility loss.}
As noted in the previous discussion (Section~\ref{sec:prelim}), minimizing the utility loss is equivalent to maximizing the utility of the top-$k$ subset induced by fair weight vectors. This optimization objective can be formalized as the following objective function by using the indicator variables:
\begin{equation}\label{eq:util-objective}
    \sum_{c\in C}\delta_c (w^o \cdot p(c)).
\end{equation}
Since $\delta_c$ encodes the top-$k$ subset membership for the candidate $c$, the summation above actually corresponds to the utility of the top-$k$ subset. Moreover, because each score under the $w^o$ can be precomputed for each candidate $c$, the coefficient associated with $\delta_c$ is fixed. Consequently, the resulting objective function is linear and can be directly incorporated into the MILP. Notably, the application of linear programming relaxation in MILP solvers may yield values of $\delta_c$ that are not exactly binary ($0$ or $1$), thereby introducing numerical errors into the objective function (\ref{eq:util-objective}). This issue can be mitigated by adjusting solver tolerances appropriately.

To obtain a stable fair weight vector, the method used in the practical $k$-level-based algorithm (Section~\ref{subsec:pract-klevel}) is also applied here. Since the method works as a post-processing step after finding a fair weight vector minimizing the utility loss, it is compatible with the MILP-based algorithm.

\section{Experiments}
In this section, we evaluate the runtime performance of the two algorithms in the two-pronged solution for $n_p > 1$ under both optimization objectives, and validate the resulting fair weight vectors. We conclude with key observations and discussions on how to choose between the $k$-level-based and MILP-based algorithms under the two optimization objectives.

\subsection{Experimental setup}\label{subsec:exp-setup}
\paragraph{Experiments Design.} In our experiments, a weight vector $w^o$ was provided as an input, along with a parameter $\epsilon$ to control allowable deviation from $w^o$ for finding a fair weight vector when $w^o$ was not fair. Formally, for a given unfair weight vector $w^o$, the goal was to find a fair weight vector $w^{\scriptscriptstyle f}$ such that $|w^{\scriptscriptstyle f}_i - w^o_i| \leq \epsilon$ for all $ 1 \leq i \leq d$ while minimizing the disparity objective $g^o(w^{\scriptscriptstyle f})$ ($w$ difference or utility loss). We conducted two set of experiments, one to compare runtime performance and another to validate the algorithm's output under the optimization objectives. Note that different methods were used to generate input weight vectors for these two sets of experiments. For the runtime comparison, weight vectors were sampled uniformly at random until $20$ \textit{unfair} weight vectors were found, and the run time was averaged over the same set of $20$ samples (excluding the time for sampling). Notably, Algorithm~\ref{alg:fairtopkver} (with engineering optimizations for the backtrack subroutine) was used to determine whether a sampled weight vector was fair. A limit of 20 hours was set on the overall run time of an algorithm, and the time for obtaining a stable fair weight vector (under the utility loss disparity objective) was omitted in our runtime comparison, as the same method were applied to all algorithms (including baselines). For the output validation, $50$ weight vectors were sampled uniformly at random and used as inputs, regardless of whether they were fair or not.

\paragraph{Datasets.} Two real-world datasets were used in our experiments. Details are given as follows:
\begin{itemize}
    \item \textbf{\textit{COMPAS}} \cite{compas}: This dataset consists of 7,214 defendants from Broward County in Florida between 2013 and 2014. It is collected and published by ProPublica for their investigation of racial bias in the criminal risk assessment software. {\tt juv\_other\_count}, {\tt c\_days\_from\_compas}, {\tt priors\_count}, {\tt start}, {\tt end} and {\tt c\_jail\_out} $-$ {\tt c\_jail\_in} (the number of days staying in jail) were used as scoring attributes (6-D), and {\tt juv\_other\_count} and {\tt c\_days\_from\_compas} were used for 2-D experiments. ``African-American'' ($51.2\%$ in the entire dataset), ``Male'' ($80.7\%$ in the entire dataset), and their intersections, ``African-American Male'' ($42.2\%$ in the entire dataset), were the protected groups.
    \item \textbf{\textit{IIT-JEE}} \cite{iitjee}: This dataset consists of scores of 384,977 students in the Mathematics, Physics, and Chemistry sections of IIT-JEE 2009 along with their genders and other personal info. All subject scores were used as scoring attributes (3-D), and the Physics and Chemistry scores were used for 2-D experiments. ``Female'' ($25.5\%$ in the entire dataset) and ``reserved-category'' \cite{baswana2019centralized} ($39.6\%$ in the entire dataset) were the protected groups, with the latter consisting students from socially and economically backward or otherwise disadvantaged sections.
\end{itemize}

As in the experiments for $n_p = 1$~\cite{cai2025finding}, all values of scoring attributes were normalized to the range of $[0, 1]$. For our multi-dimensional experiments ($d \geq 3$), datasets were also preprocessed to shrink the input size using the $k$-skyband~\cite{papadias2005progressive, sheng2012worst} and the reverse top-$k$ query~\cite{vlachou2010reverse, chen2023not}, as in the $n_p = 1$ experiments. Since the size of the IIT-JEE dataset was larger, we used larger values of $k$ for it (ranging from $50$ to $500$), while smaller values of $k$ were used for the COMPAS dataset (ranging from $10$ to $100$). 

Synthetic data were deliberately excluded due to challenges in ``meaningfully'' reflecting real-world bias, as naive synthetic datasets (same distribution for all groups) diminish the need for fair top-$k$ selection.

\paragraph*{Baselines.} For comparison, \textsc{2draysweep} and \textsc{ATC\textsubscript{$+$}} from \cite{asudeh2019designing} were implemented as baselines, with the former (\textsc{2draysweep}) for 2-D experiments and the latter (\textsc{ATC\textsubscript{$+$}}) for multi-dimensional experiments. The improvements that enhanced their performance for $n_p = 1$ were also applied (as in~\cite{cai2025finding}), along with additional refinements regarding the size of $V$ (see Section~\ref{subsubsec:pract-klevel-md}) incorporated into \textsc{ATC\textsubscript{$+$}}. The tie-breaking was handled in the implementations of \textsc{2draysweep} but was ignored in the implementation of \textsc{ATC\textsubscript{$+$}}, as it was not obvious how to do so without introducing significant runtime overhead for \textsc{ATC\textsubscript{$+$}}.

\paragraph*{Environment and Implementation.} All our implementations were in C\texttt{++} on a 64-core AMD EPYC 7763 (2.45 GHz) machine with 512GB RAM, and our code\footnote{Available at \url{https://github.com/caiguangya/fair-topk-general}.\label{repo}} was compiled and run inside an Apptainer \cite{kurtzer2017singularity} container\footref{repo} with a Rocky Linux 8.10 host OS.

\paragraph*{Default values.} For experiments in 2-D, $k = 50$ and $\epsilon = 0.1$. For experiments in higher dimensions ($d \geq 3$), $k = 50$ and $\epsilon = 0.05$. For the fairness constraints, the proportion of the protected group candidates in top-$k$ selection results was constrained to $[40\%, 60\%]$, $[70\%, 90\%]$ and $[30\%, 55\%]$ for protected groups of COMPAS (``African-American'', ``Male'' and ``African-American Male'', respectively), and was constrained to $[10\%, 40\%]$ and $[15\%, 65\%]$ for protected groups of IIT-JEE (``Female'' and ``reserved-category'', respectively).

\subsection{Runtime experiments}
Since the main focus of this work is the speed of execution, we first present the results of the runtime experiments, starting with the 2-D experiments followed by results in higher dimensions.

\subsubsection{Runtime experiments for 2-D datasets}\label{subsubsec:exp2d-multi}
\begin{figure}[tbh!]
    \captionsetup[subfigure]{justification=centering}
    \centering
\begin{subfigure}[b]{0.3\columnwidth}
    \resizebox{\columnwidth}{!}{\begin{tikzpicture}
    \begin{axis}[
        width =\axisdefaultwidth,
        height= 180pt,
        legend style={font=\normalsize},
        legend pos=south east,
        xlabel={\LARGE $k$},
        ylabel={\Large Time (s)},
        xtick=data,
        xticklabels={10, 20, 50, 70, 100},
        xticklabel style = {font=\large},
        yticklabel style = {font=\large},
        mark size=4pt,
        ymode=log,
        ymin=1e-4,
        ymax=4e-3]
        \addplot[mark=o, color=blue] table [x=k, y=k-level-based-wd, col sep=comma] {compas-2d-k-multi.csv};
        \addplot[mark=o, mark options=solid, color=blue, densely dashed] table [x=k, y=k-level-based-u, col sep=comma] {compas-2d-k-multi.csv};
        \addplot[mark=diamond, color=black] table [x=k, y=baseline-wd, col sep=comma] {compas-2d-k-multi.csv};
        \addplot[mark=diamond, mark options=solid, color=black, densely dashed] table [x=k, y=baseline-u, col sep=comma] {compas-2d-k-multi.csv};
    \end{axis}
\end{tikzpicture}}
\caption{COMPAS, varying $k$}\label{2d-k1-multi}
\end{subfigure}
\hspace{10pt}
\begin{subfigure}[b]{0.3\columnwidth}
    \resizebox{\columnwidth}{!}{\begin{tikzpicture}
        \begin{axis}[
            width =\axisdefaultwidth,
            height= 180pt,
            xlabel={\LARGE $k$},
            ylabel={\Large Time (s)},
            xtick=data,
            xticklabels={50, 70, 100, 200, 500},
            xticklabel style = {font=\large},
            yticklabel style = {font=\large},
            mark size=4pt,
            ymode=log,
            ymin=0.009,
            ymax=1]
            \addplot[mark=o, color=blue] table [x=k, y=k-level-based-wd, col sep=comma] {jee-2d-k-multi.csv};
            \addplot[mark=o, mark options=solid, color=blue, densely dashed] table [x=k, y=k-level-based-u, col sep=comma] {jee-2d-k-multi.csv};
            \addplot[mark=diamond, color=black] table [x=k, y=baseline-wd, col sep=comma] {jee-2d-k-multi.csv};
            \addplot[mark=diamond, mark options=solid, color=black, densely dashed] table [x=k, y=baseline-u, col sep=comma] {jee-2d-k-multi.csv};
        \end{axis}
    \end{tikzpicture}}
    \caption{IIT-JEE, varying $k$}\label{2d-k2-multi}
\end{subfigure}
\hspace{10pt}
    \begin{subfigure}[b]{0.3\columnwidth}
        \resizebox{\columnwidth}{!}{\begin{tikzpicture}
            \begin{axis}[hide axis, legend columns=1, legend style={at={(0.58, 0.85)},anchor=north,legend cell align={left}}]
            
            \addlegendimage{mark=o, color=blue, mark size=4pt};
            \addlegendentry{$k$-level-based ($w$ difference)}
            
            \addlegendimage{mark=o, mark options=solid, densely dashed, color=blue, mark size=4pt};
            \addlegendentry{$k$-level-based (Utility loss)}
            
            \addlegendimage{mark=diamond, color=black, mark size=4pt};
            \addlegendentry{Baseline ($w$ difference)}
            
            \addlegendimage{mark=diamond, mark options=solid, densely dashed, color=black, mark size=4pt};
            \addlegendentry{Baseline (Utility loss)}

            \addplot [draw=none, forget plot] coordinates {(0,0)};
            \end{axis}
         \end{tikzpicture}}
    \end{subfigure}

\begin{subfigure}[b]{0.3\columnwidth}
    \resizebox{\columnwidth}{!}{\begin{tikzpicture}
    \begin{axis}[
        width =\axisdefaultwidth,
        height= 180pt,
        xlabel={\LARGE $\epsilon$},
        ylabel={\Large Time (s)},
        xtick=data,
        xticklabels={0.01, 0.05, 0.1, 0.15, 0.2},
        xticklabel style = {font=\large},
        yticklabel style = {font=\large},
        mark size=4pt,
        ymode=log,
        ymin=1e-4,
        ymax=4e-3]
        \addplot[mark=o, color=blue] table [x=epsilon, y=k-level-based-wd, col sep=comma] {compas-2d-eps-multi.csv};
        \addplot[mark=o, mark options=solid, color=blue, densely dashed] table [x=epsilon, y=k-level-based-u, col sep=comma] {compas-2d-eps-multi.csv};
        \addplot[mark=diamond, color=black] table [x=epsilon, y=baseline-wd, col sep=comma] {compas-2d-eps-multi.csv};
        \addplot[mark=diamond, mark options=solid, color=black, densely dashed] table [x=epsilon, y=baseline-u, col sep=comma] {compas-2d-eps-multi.csv};
    \end{axis}
\end{tikzpicture}}
\caption{COMPAS, varying $\epsilon$}\label{2d-eps1-multi}
\end{subfigure}
\hspace{10pt}
\begin{subfigure}[b]{0.3\columnwidth}
    \resizebox{\columnwidth}{!}{\begin{tikzpicture}
        \begin{axis}[
            width =\axisdefaultwidth,
            height= 180pt,
            xlabel={\LARGE $\epsilon$},
            ylabel={\Large Time (s)},
            xtick=data,
            xticklabels={0.01, 0.05, 0.1, 0.15, 0.2},
            xticklabel style = {font=\large},
            yticklabel style = {font=\large},
            mark size=4pt,
            ymode=log,
            ymin=0.009,
            ymax=1.9]
            \addplot[mark=o, color=blue] table [x=epsilon, y=k-level-based-wd, col sep=comma] {jee-2d-eps-multi.csv};
            \addplot[mark=o, mark options=solid, color=blue, densely dashed] table [x=epsilon, y=k-level-based-u, col sep=comma] {jee-2d-eps-multi.csv};
            \addplot[mark=diamond, color=black] table [x=epsilon, y=baseline-wd, col sep=comma] {jee-2d-eps-multi.csv};
            \addplot[mark=diamond, mark options=solid, color=black, densely dashed] table [x=epsilon, y=baseline-u, col sep=comma] {jee-2d-eps-multi.csv};
        \end{axis}
    \end{tikzpicture}}
    \caption{IIT-JEE, varying $\epsilon$}\label{2d-eps2-multi}
\end{subfigure}
\hspace{10pt}
\begin{subfigure}[b]{0.3\columnwidth}
    \resizebox{\columnwidth}{!}{\begin{tikzpicture}
        \begin{axis}[
            width =\axisdefaultwidth,
            height= 180pt,
            xlabel={\LARGE $n$ (ratio of full size)},
            ylabel={\Large Time (s)},
            xtick=data,
            xticklabels={0.2, 0.4, 0.6, 0.8, 1.0},
            xticklabel style = {font=\large},
            yticklabel style = {font=\large},
            mark size=4pt,
            ymode=log,
            ymin=0.001,
            ymax=1.0]
            \addplot[mark=o, color=blue] table [x=epsilon, y=k-level-based-wd, col sep=comma] {jee-2d-n-multi.csv};
            \addplot[mark=o, mark options=solid, color=blue, densely dashed] table [x=epsilon, y=k-level-based-u, col sep=comma] {jee-2d-n-multi.csv};
            \addplot[mark=diamond, color=black] table [x=epsilon, y=baseline-wd, col sep=comma] {jee-2d-n-multi.csv};
            \addplot[mark=diamond, mark options=solid, color=black, densely dashed] table [x=epsilon, y=baseline-u, col sep=comma] {jee-2d-n-multi.csv};
        \end{axis}
    \end{tikzpicture}}
    \caption{IIT-JEE, varying $n$}\label{2d-n-multi}
\end{subfigure}
\caption{Runtime experimental results for 2-D datasets.}
\end{figure}
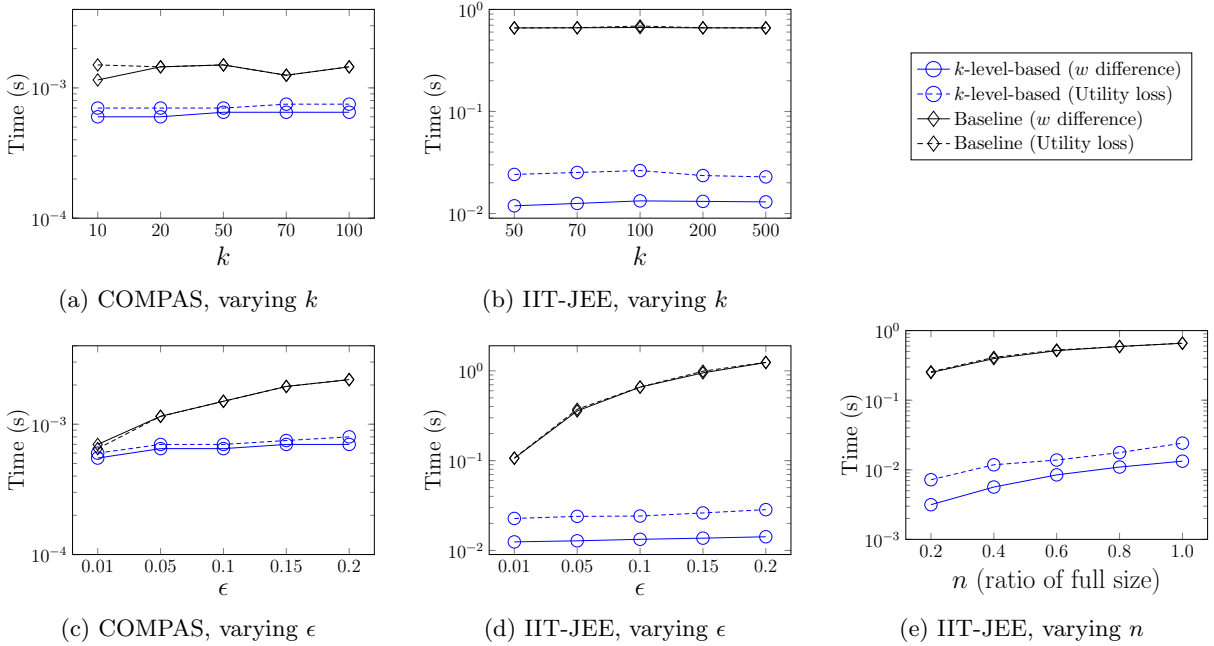

\paragraph{Varying $\boldsymbol{k}$.} Figures~\ref{2d-k1-multi} and \ref{2d-k2-multi} show the average run times as $k$ varies. The $k$-level-based algorithm performed consistently better than the baseline algorithm, achieving speedups of up to $50$x and $28$x for minimizing the $w$ difference and the utility loss, respectively. The performance of the $k$-level-based algorithm remained stable as $k$ increased since the actual $k$-level structural complexity was low. Consequently, the sweeping process accounted for only a small fraction of the total run time, as also noted in the analysis of $n_p = 1$ experiments~\cite{cai2025finding}.

\paragraph{Varying $\boldsymbol{\epsilon}$.} Figures~\ref{2d-eps1-multi} and \ref{2d-eps2-multi} show the average run times as $\epsilon$ varies. The $k$-level-based algorithm still performed consistently better than the baseline algorithm. The performance trend with respect to $\epsilon$ was consistent with the observations from $n_p = 1$ experiments.

\paragraph{Varying $\boldsymbol{n}$.} Figure~\ref{2d-n-multi} shows average run times as $n$ varies. We controlled $n$ by randomly selecting $20\%$, $40\%$, $60\%$, $80\%$, and $100\%$ of the IIT-JEE dataset. The $k$-level-based algorithm again consistently outperformed the baseline algorithm in all cases, with the performance trend being consistent with the observations from $n_p = 1$ experiments.

\paragraph{$\boldsymbol{w}$ difference vs. utility loss.} Across all experiments, utility loss minimization incurred greater overhead than $w$ difference minimization in the $k$-level-based algorithm. This increase stemmed from additional steps required to compute the utility of a top-$k$ subset in the backtracking algorithm and the cost of prerequisite computations (see Section~\ref{subsec:backtrack}). Nevertheless, the $k$-level-based algorithm with utility loss minimization remained efficient compared to the baseline algorithm. While utility loss minimization introduced similar overhead to the baseline algorithm, it accounted for only a small fraction of the total run time, given the baseline algorithm's inherently long execution time. Therefore, the performance of the baseline algorithm was nearly identical under the two optimization objectives.

\subsubsection{Runtime experiments for multi-dimensional datasets}\label{subsubsec:expmd-multi}
\begin{figure}[tbh!]
    \captionsetup[subfigure]{justification=centering}
    \centering
\begin{subfigure}[b]{0.3\columnwidth}
    \resizebox{\columnwidth}{!}{\begin{tikzpicture}
        \begin{axis}[
            width =\axisdefaultwidth,
            height= 180pt,
            legend pos=north east,
            xtick=data,
            symbolic x coords={COMPAS-wd, COMPAS-u, IIT-JEE-wd, IIT-JEE-u},
            xticklabels={\small COMPAS\\($w$ diff.), \small COMPAS\\(Util. loss), \small IIT-JEE\\($w$ diff.),  \small IIT-JEE\\(Util. loss)},
            x tick label style={
            align=center,
            },
            enlarge x limits= 0.2,
            bar width=10,
            ylabel={\large Time (s)},
            yticklabel style = {font=\Large},
            log origin=infty,
            ybar,
            ymin=1e-3,
            ytick={1e-3, 1e-2, 1e-1, 1, 10, 100, 1000, 10000},
            ymode=log]
            \addplot[red, pattern=north east lines, pattern color=red] coordinates {(COMPAS-wd, 1.515e-02)  (COMPAS-u, 1.240e-02) (IIT-JEE-wd, 3.485e-02)  (IIT-JEE-u, 2.990e-02)};
            \addplot[blue, pattern=north west lines, pattern color=blue] coordinates {(COMPAS-wd, 1.155e-02) (COMPAS-u, 1.090e-02) (IIT-JEE-wd, 1.930e-02) (IIT-JEE-u, 1.815e-02)};
            \addplot [pattern=crosshatch] coordinates {(COMPAS-wd, 1.860e+03) (COMPAS-u, 1.858e+03) (IIT-JEE-wd, 1.637e+01)  (IIT-JEE-u, 1.659e+01)};
            \legend{MILP-based, $k$-level-based, Baseline}
        \end{axis}
\end{tikzpicture}}
\caption{Baseline comparison}\label{md-baseline-compare-multi}
\end{subfigure}
\hspace{10pt}
\begin{subfigure}[b]{0.3\columnwidth}
    \resizebox{\columnwidth}{!}{\begin{tikzpicture}
    \begin{axis}[
        width =\axisdefaultwidth,
        height= 180pt,
        xlabel={\LARGE $k$},
        ylabel={\large Time (s)},
        xtick=data,
        xticklabels={10, 20, 50, 70, 100},
        yticklabel style = {font=\large},
        xticklabel style = {font=\large},
        unbounded coords=jump,
        mark size=4pt,
        ymode=log]
        \addplot[mark=square, color=red] table [x=k, y=milp-based-wd, col sep=comma] {compas-k-multi.csv};
        \addplot[mark=o, color=blue] table [x=k, y=k-level-based-wd, col sep=comma] {compas-k-multi.csv};
        \addplot[mark=square, color=red, mark options=solid, densely dashed] table [x=k, y=milp-based-u, col sep=comma] {compas-k-multi.csv};
        \addplot[mark=o, color=blue, mark options=solid, densely dashed] table [x=k, y=k-level-based-u, col sep=comma] {compas-k-multi.csv};
    \end{axis}
\end{tikzpicture}}
\caption{COMPAS, varying $k$}\label{hd-k1-multi}
\end{subfigure}
\hspace{10pt}
\begin{subfigure}[b]{0.3\columnwidth}
    \resizebox{\columnwidth}{!}{\begin{tikzpicture}
        \begin{axis}[
            width =\axisdefaultwidth,
            height= 180pt,
            xlabel={\LARGE $k$},
            ylabel={\large Time (s)},
            xtick=data,
            xticklabels={50, 70, 100, 200, 500},
            yticklabel style = {font=\large},
            xticklabel style = {font=\large},
            mark size=4pt,
            ymode=log]
            \addplot[mark=square, color=red] table [x=k, y=milp-based-wd, col sep=comma] {jee-k-multi.csv};
            \addplot[mark=o, color=blue] table [x=k, y=k-level-based-wd, col sep=comma] {jee-k-multi.csv};
            \addplot[mark=square, color=red, mark options=solid, densely dashed] table [x=k, y=milp-based-u, col sep=comma] {jee-k-multi.csv};
            \addplot[mark=o, color=blue, mark options=solid, densely dashed] table [x=k, y=k-level-based-u, col sep=comma] {jee-k-multi.csv};
        \end{axis}
    \end{tikzpicture}}
    \caption{IIT-JEE, varying $k$}\label{hd-k2-multi}
\end{subfigure}

\begin{subfigure}[b]{0.3\columnwidth}
    \resizebox{\columnwidth}{!}{\begin{tikzpicture}
    \begin{axis}[
        width =\axisdefaultwidth,
        height= 180pt,
        xlabel={\LARGE $\epsilon$},
        ylabel={\large Time (s)},
        xtick=data,
        xticklabels={0.01, 0.025, 0.05, 0.075, 0.1},
        yticklabel style = {font=\large},
        xticklabel style = {font=\large},
        mark size=3pt,
        ymode=log]
        \addplot[mark=square, color=red] table [x=epsilon, y=milp-based-wd, col sep=comma] {compas-eps-multi.csv};
        \addplot[mark=o, color=blue] table [x=epsilon, y=k-level-based-wd, col sep=comma] {compas-eps-multi.csv};
        \addplot[mark=square, color=red, mark options=solid, densely dashed] table [x=epsilon, y=milp-based-u, col sep=comma] {compas-eps-multi.csv};
        \addplot[mark=o, color=blue, mark options=solid, densely dashed] table [x=epsilon, y=k-level-based-u, col sep=comma] {compas-eps-multi.csv};
    \end{axis}
\end{tikzpicture}}
\caption{COMPAS, varying $\epsilon$}\label{hd-eps1-multi}
\end{subfigure}
\hspace{10pt}
\begin{subfigure}[b]{0.3\columnwidth}
    \resizebox{\columnwidth}{!}{\begin{tikzpicture}
        \begin{axis}[
            width =\axisdefaultwidth,
            height= 180pt,
            xlabel={\LARGE $\epsilon$},
            ylabel={\large Time (s)},
            xtick=data,
            xticklabels={0.01, 0.025, 0.05, 0.075, 0.1},
            yticklabel style = {font=\large},
            xticklabel style = {font=\large},
            mark size=3pt,
            ymode=log,
            ymax=1.0]
        \addplot[mark=square, color=red] table [x=epsilon, y=milp-based-wd, col sep=comma] {jee-eps-multi.csv};
        \addplot[mark=o, color=blue] table [x=epsilon, y=k-level-based-wd, col sep=comma] {jee-eps-multi.csv};
        \addplot[mark=square, color=red, mark options=solid, densely dashed] table [x=epsilon, y=milp-based-u, col sep=comma] {jee-eps-multi.csv};
        \addplot[mark=o, color=blue, mark options=solid, densely dashed] table [x=epsilon, y=k-level-based-u, col sep=comma] {jee-eps-multi.csv};
        \end{axis}
    \end{tikzpicture}}
    \caption{IIT-JEE, varying $\epsilon$}\label{hd-eps2-multi}
\end{subfigure}
\hspace{10pt}
\begin{subfigure}[b]{0.3\columnwidth}
    \resizebox{\columnwidth}{!}{\begin{tikzpicture}
        \begin{axis}[
            width =\axisdefaultwidth,
            height= 180pt,
            xlabel={\LARGE $d$},
            ylabel={\large Time (s)},
            xtick=data,
            xticklabels={3, 4, 5, 6},
            yticklabel style = {font=\large},
            xticklabel style = {font=\large},
            mark size=4pt,
            ymode=log]
        \addplot[mark=square, color=red] table [x=d, y=milp-based-wd, col sep=comma] {compas-d-multi.csv};
        \addplot[mark=o, color=blue] table [x=d, y=k-level-based-wd, col sep=comma] {compas-d-multi.csv};
        \addplot[mark=square, color=red, mark options=solid, densely dashed] table [x=d, y=milp-based-u, col sep=comma] {compas-d-multi.csv};
        \addplot[mark=o, color=blue, mark options=solid, densely dashed] table [x=d, y=k-level-based-u, col sep=comma] {compas-d-multi.csv};
        \end{axis}
    \end{tikzpicture}}
\caption{COMPAS, varying $d$}\label{hd-d-multi}
\end{subfigure}

\begin{subfigure}[b]{0.3\columnwidth}
    \resizebox{\columnwidth}{!}{\begin{tikzpicture}
        \begin{axis}[
            width =\axisdefaultwidth,
            height= 180pt,
            xlabel={\LARGE $n$ (ratio of full size)},
            ylabel={\large Time (s)},
            xtick=data,
            xticklabels={0.2, 0.4, 0.6, 0.8, 1.0},
            yticklabel style = {font=\large},
            xticklabel style = {font=\large},
            mark size=4pt,
            ymode=log]
            \addplot[mark=square, color=red] table [x=n, y=milp-based-wd, col sep=comma] {compas-n-multi.csv};
        \addplot[mark=o, color=blue] table [x=n, y=k-level-based-wd, col sep=comma] {compas-n-multi.csv};
        \addplot[mark=square, color=red, mark options=solid, densely dashed] table [x=n, y=milp-based-u, col sep=comma] {compas-n-multi.csv};
        \addplot[mark=o, color=blue, mark options=solid, densely dashed] table [x=n, y=k-level-based-u, col sep=comma] {compas-n-multi.csv};
        \end{axis}
    \end{tikzpicture}}
    \caption{COMPAS, varying $n$}\label{hd-n1-multi}
\end{subfigure}
\hspace{10pt}
\begin{subfigure}[b]{0.3\columnwidth}
    \resizebox{\columnwidth}{!}{\begin{tikzpicture}
        \begin{axis}[
            width =\axisdefaultwidth,
            height= 180pt,
            legend style={font=\small},
            legend pos=south east,
            xlabel={\LARGE $n$ (ratio of full size)},
            ylabel={\large Time (s)},
            xtick=data,
            xticklabels={0.2, 0.4, 0.6, 0.8, 1.0},
            yticklabel style = {font=\large},
            xticklabel style = {font=\large},
            mark size=4pt,
            ymode=log,
            ymax=1.0]
        \addplot[mark=square, color=red] table [x=n, y=milp-based-wd, col sep=comma] {jee-n-multi.csv};
        \addplot[mark=o, color=blue] table [x=n, y=k-level-based-wd, col sep=comma] {jee-n-multi.csv};
        \addplot[mark=square, color=red, mark options=solid, densely dashed] table [x=n, y=milp-based-u, col sep=comma] {jee-n-multi.csv};
        \addplot[mark=o, color=blue, mark options=solid, densely dashed] table [x=n, y=k-level-based-u, col sep=comma] {jee-n-multi.csv};
        \end{axis}
    \end{tikzpicture}}
\caption{IIT-JEE, varying $n$}\label{hd-n2-multi}
\end{subfigure}
\hspace{10pt}
\begin{subfigure}[b]{0.3\columnwidth}
        \resizebox{\columnwidth}{!}{\begin{tikzpicture}
            \begin{axis}[hide axis, legend columns=1, legend style={at={(0.58, 0.85)},anchor=north}, legend cell align={left}]
            
            \addlegendimage{mark=square, color=red, mark size=4pt};
            \addlegendentry{MILP-based ($w$ difference)}
            
            \addlegendimage{mark=square, color=red, mark options=solid, densely dashed, mark size=4pt};
            \addlegendentry{MILP-based (Utility loss)}

            \addlegendimage{mark=o, color=blue, mark size=4pt};
            \addlegendentry{$k$-level-based ($w$ difference)}
            
            \addlegendimage{mark=o, mark options=solid, densely dashed, color=blue, mark size=4pt};
            \addlegendentry{$k$-level-based (Utility loss)}

            \addplot [draw=none, forget plot] coordinates {(0,0)};
            \end{axis}
         \end{tikzpicture}}
    \end{subfigure}
\caption{Runtime experimental results for multi-dimensional datasets ($3 \leq d \leq 6$).}
\end{figure}
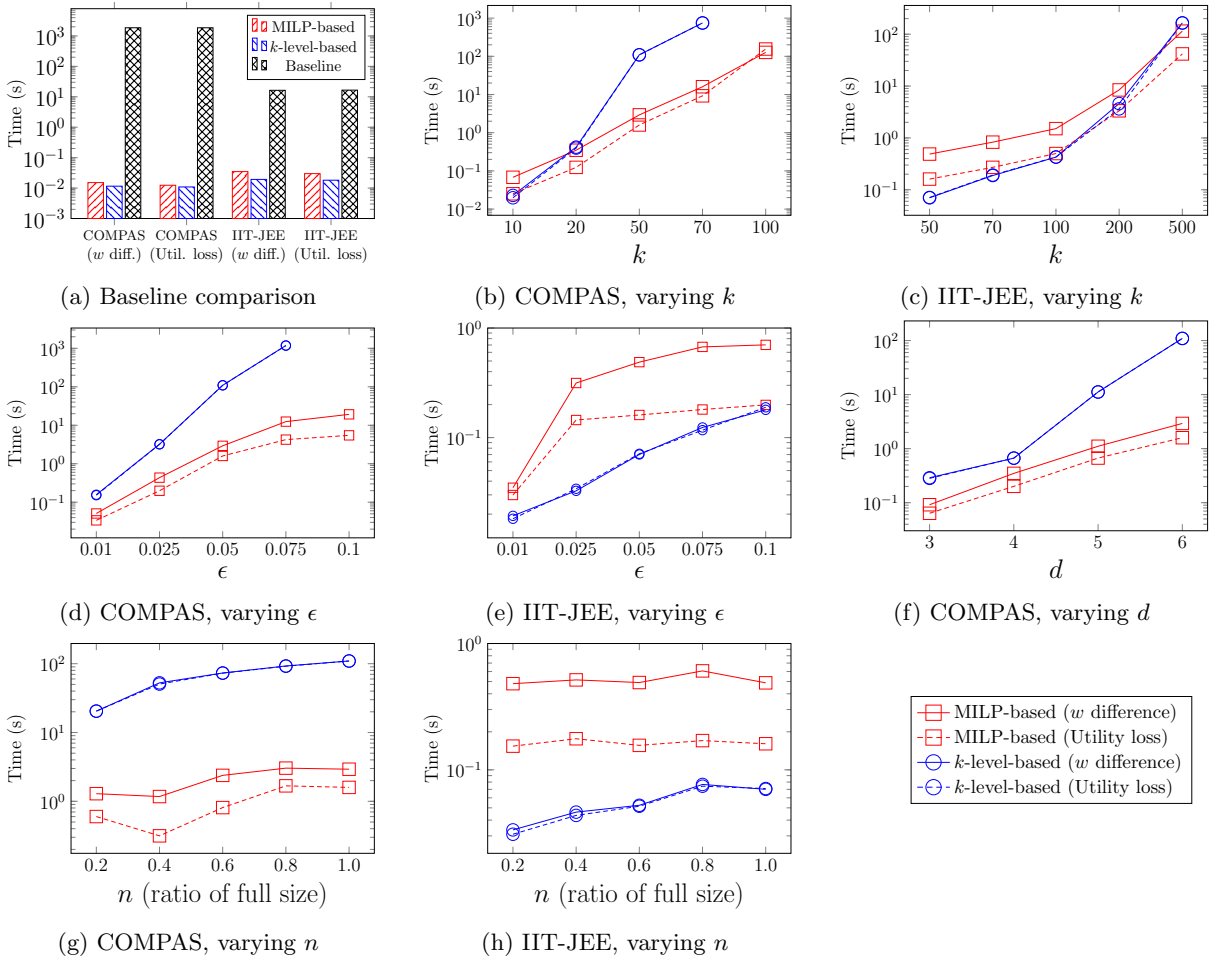

\paragraph{Baseline comparison.} Figure~\ref{md-baseline-compare-multi} shows a comparison of the two algorithms in the two-pronged solution against the baseline. For the COMPAS dataset, we set $k = 10$ and $\epsilon = 0.001$, while for the IIT-JEE dataset, we used $k = 50$ and $\epsilon = 0.01$. Despite the optimization applied to the $k$-level-based algorithm regarding the size of $V$ (Section~\ref{subsubsec:pract-klevel-md}) also improved the runtime efficiency of the baseline, particularly for the lower dimensional dataset (IIT-JEE), our two algorithms still achieved speedups of several orders of magnitude.

\paragraph{Varying $\boldsymbol{k}$.} Figures~\ref{hd-k1-multi} and \ref{hd-k2-multi} shows the average run times as $k$ varies. The $k$-level-based algorithm outperformed the MILP-based algorithm when $k$ was small, especially for the lower dimensional dataset (IIT-JEE). For the higher dimensional dataset (COMPAS), The MILP-based algorithm performed better. Despite the runtime improvement on the $k$-level-based algorithm due to the optimization regarding the size of $V$ (Section~\ref{subsubsec:pract-klevel-md}), it still failed to finish within the time limit when $k=100$. The general performance trends of the algorithms were consistent with the observations from $n_p = 1$ experiments~\cite{cai2025finding}.

\paragraph{Varying $\boldsymbol{\epsilon}$.} Figures~\ref{hd-eps1-multi} and \ref{hd-eps2-multi} show the average run times as $\epsilon$ varies. In both datasets, the run times of both algorithms increased as $\epsilon$ increased, which were consistent with the observations from $n_p = 1$ experiments.

\paragraph{Varying $\boldsymbol{n}$.} Figures~\ref{hd-n1-multi} and \ref{hd-n2-multi} show the average run times as $n$ varies. We controlled $n$ by randomly selecting a subset of data before preprocessing. In general, the run times of both algorithms increased as $n$ increased. For the $k$-level algorithm, its run time showed greater dependence on $n$ compared to experimental results with $n_p = 1$, as it cannot terminate early given the optimization objective. Some anomalies still existed but were also consistent with the observations from $n_p = 1$ experiments.

\paragraph{Varying $\boldsymbol{d}$.} Figure~\ref{hd-d-multi} shows the average run times as $d$ varies. We controlled $d$ by selecting a subset of scoring attributes of the COMPAS dataset after preprocessing. The run times of both algorithm increased dramatically as the dimensionality $d$ increased, with the $k$-level-algorithm increasing more rapidly, partly due to the linear programming algorithm~\cite{seidel1991small} used in our implementation. The general performance trends of the algorithms are also consistent with the observations from $n_p = 1$ experiments.

\paragraph{$\boldsymbol{w}$ difference vs. utility loss.} For the $k$-level based algorithm, minimizing the $w$ difference is expected to incur greater overhead than minimizing the utility loss, as the former requires solving a complex linear program while the latter only involves summing the scores of $k$ elements. However, their empirical performance was nearly identical across all experiments. Since this additional overhead applies only to fair cells, this observation suggests that the number of fair cells was small such that the extra computational costs incurred by fair cells may be negligible compared to the time spent on $(k-1)$-level cell traversal. In contrast, for the MILP-based algorithm, minimizing the $w$ difference does introduce noticeable overhead compared to minimizing the utility loss, with the former generally being slower than the latter.

\subsection{Validation experiments}
In validation experiments, we show that our algorithms indeed minimize the $w$ difference or the utility loss, in contrast to prior unaugmented algorithms~\cite{cai2025finding} that return an arbitrary fair weight vector within $V$. Notably, a variant of our algorithm for the Fair Top-$k$ Verification problem (Algorithm~\ref{alg:fairtopkver}) was used to compute utility losses, which was one of the reasons for making its corresponding backtracking subroutine numerically more robust (Section~\ref{subsec:backtrack}). Table~\ref{tab:validatation} compares the optimal $w$ differences and utility losses obtained by our augmented $k$-level-based and MILP-based algorithms with those reported in~\cite{cai2025finding}, which were obtained using unaugmented ones that did not optimize these objectives. Both our augmented algorithms achieved the same optimal results. Compared to the (unaugmented) MILP-based algorithm, improvements in terms of $w$ difference or utility loss were clear. Compared to the (unaugmented) $k$-level-based algorithm, the improvement in the $w$ difference was evident; however, the improvement in the utility loss was relatively modest, suggesting its breadth-search approach is effective in searching a fair weight vector with a small utility loss. Results consistent with the $n_p=1$ experiments were also observed for $n_p > 1$, as reported in Table~\ref{tab:validatation2}.

\begin{table}[tbh!]
  \centering
  \caption{Validation results for $n_p = 1$.}\label{tab:validatation}
  \begin{threeparttable}
  \begin{tabular}{ |c|c|c|ccc|ccc| }
    \toprule
       \multirow{2}{*}{Dataset} & \multirow{2}{*}{$\epsilon$} & \multirow{2}{*}{\shortstack{Fairness \\ constraint\tnote{1}}} & \multicolumn{3}{c|}{$w$ difference}  & \multicolumn{3}{c|}{Utility loss} \\
       & & & $k$-level & MILP & Optimal & $k$-level & MILP & Optimal \\
    \hline
       \multirow{6}{*}{\shortstack{COMPAS\\($k=50$)}} & \multirow{3}{*}{0.05} & $[34\%, 66\%]$ & 0.149 & 0.211 & \textbf{0.060} & 0.476\% & 1.210\% & \textbf{0.392\%}\\
      & & $[40\%, 60\%]$ & 0.144 & 0.213 & \textbf{0.073} & 0.449\% & 1.328\% & \textbf{0.347\%}\\ 
      & & $[44\%, 56\%]$ & 0.155 & 0.202 & \textbf{0.086} & 0.718\% & 1.319\% & \textbf{0.656\%}\\ 
    \cline{2-9}
      & 0.025 & \multirow{3}{*}{$[40\%, 60\%]$} & 0.088 & 0.090 & \textbf{0.044} & 0.121\% & 0.142\% & \textbf{0.086\%} \\
      & 0.05 & & 0.144 & 0.213 & \textbf{0.073} & 0.449\% & 1.328\% & \textbf{0.347\%} \\ 
      & 0.075 & & 0.247 & 0.308 & \textbf{0.162} & 2.416\% & 3.187\% & \textbf{2.161\%} \\ 
    \midrule
    \multirow{6}{*}{\shortstack{IIT-JEE\\($k=100$)}} & \multirow{3}{*}{0.05} & $[5\%, 45\%]$ & 0.055 & 0.089 & \textbf{0.039} & 0.033\% & 0.067\% & \textbf{0.031\%}\\
    & & $[7\%, 43\%]$ & 0.062 & 0.089 & \textbf{0.055} & 0.025\% & 0.043\% & \textbf{0.019\%}\\ 
    & & $[9\%, 41\%]$ & 0.089 & 0.092 & \textbf{0.072} & 0.028\% & 0.027\% & \textbf{0.021\%}\\
    \cline{2-9}
    & 0.025 & \multirow{3}{*}{$[9\%, 41\%]$} & 0.050 & 0.050 & \textbf{0.039} & 0.022\% & 0.020\% & \textbf{0.020\%}\\
    & 0.05 & & 0.089 & 0.092 & \textbf{0.072} & 0.028\% & 0.027\% & \textbf{0.021\%} \\ 
    & 0.075 & & 0.102 & 0.134 & \textbf{0.083} & 0.036\% & 0.096\% & \textbf{0.031\%} \\
    \bottomrule
  \end{tabular}
  \begin{tablenotes}
    \footnotesize
    \item[1]The fairness constraint for ``African-American'' (COMPAS) and ``Female'' (IIT-JEE), respectively.
  \end{tablenotes}
  \end{threeparttable}
\end{table}

\begin{table}[tbh!]
  \centering
  \caption{Validation results for COMPAS dataset with default fairness constraints ($n_p = 3$) and $k = 50$.}\label{tab:validatation2}
  \begin{threeparttable}
  \begin{tabular}{ |c|c|ccc|ccc| }
    \toprule
       \multirow{2}{*}{$\epsilon$} & \multirow{2}{*}{\shortstack{Found/Unfair\\Ratio\tnote{1}}} & \multicolumn{3}{c|}{$w$ difference}  & \multicolumn{3}{c|}{Utility loss} \\
       & & $k$-level & MILP & Optimal & $k$-level & MILP & Optimal \\
    \hline
      0.025 & 2/45 & 0.068 & 0.133 & \textbf{0.025} & 0.075\% & 0.310\% & \textbf{0.032\%} \\
      0.05 & 7/45 & 0.184 & 0.203 & \textbf{0.113} & 0.093\% & 1.227\% & \textbf{0.080\%} \\ 
      0.075 & 16/45 & 0.273 & 0.300 & \textbf{0.191} & 2.635\% & 3.338\% & \textbf{2.458\%} \\ 
    \bottomrule
  \end{tabular}
  \begin{tablenotes}
    \footnotesize
    \item[1]The ratio between the number of input unfair weight vectors and found fair ones.
  \end{tablenotes}
  \end{threeparttable}
\end{table}

\subsection{Observation summary and algorithm selection}
In this section, we summarize key experimental observations and discuss algorithm selection regarding the two optimization objectives based on them.

\paragraph{2-D experiments.} For the $k$-level-based algorithm, the performance trends seen by varying difference variables ($k$, $\epsilon$ and $n$) were consistent with the previous observations from $n_p = 1$ experiments, which evaluated an unaugmented algorithm without an optimization objective. The run times also remained stable across $k$ for both datasets, as the sweeping process accounted for only a small fraction of the total run time. This reinforced our choice of the unidirectional sweeping approach in the 2-D $k$-level-based algorithm. The utility loss minimization incurred greater overhead than the $w$ difference minimization, partly because the backtracking subroutine for utility loss minimization was engineered more conservatively. The same subroutine was also used in the validation experiments, raising the importance of ensuring its correctness and accuracy. Nevertheless, the (augmented) $k$-level-based algorithm was still substantial faster than the baseline for minimizing the utility loss.

\paragraph{Multi-dimensional experiments.} The performance trends of both the $k$-level-based and the MILP-based algorithm were also  consistent with the observations from $n_p = 1$ experiments, which evaluated unaugmented algorithms without an optimization objective. Unlike the 2-D case, the run times of the $k$-level-based algorithm under the two optimization objectives were nearly identical across all experiment settings. This was because the additional overhead applies only to fair cells, which were few in number in our experiments. In contrast, the MILP-based algorithm was noticeably slower in minimizing the $w$ difference than in minimizing the utility loss.

Regarding selecting between the two algorithms of the two-pronged solution, the general guideline, that one selects the $k$-level-based algorithm for a small $k$ and the MILP-based algorithm for a large $k$, remains valid. The threshold value of $k$ for choosing between the two algorithms still decreases as the dimensionality $d$ increases. However, this threshold is now influenced by the choice of optimization objective, as it is likely to be higher when minimizing the $w$ difference than when minimizing the utility loss.

\section{Conclusion}
We presented an integrated study of the fair top-$k$ selection problem in a generalized setting that considers multiple protected groups while also minimizes the disparity from a reference, unfair scoring function. We started with the hardness analysis---driven by the necessities of experimental exploration---to establish the theoretical foundation and guide algorithm design. This was followed by algorithm design exploiting the gap revealed by the hardness analysis and augmenting the two-pronged solution. We then addressed practical engineering considerations, balancing implementation complexity, robustness, and performance through careful engineering trade-offs. Finally, empirical evaluation explored real-world scenarios and assessed the impact of the newly introduced problem constraints, informing algorithm design and implementation decisions. Through this integrative framework that integrates hardness analysis, algorithm design, practical engineering and empirical evaluation, we arrive at an efficient two-pronged solution that is significantly faster than the state of the art, accompanied by many insights into problem structure and algorithmic behavior.

\section*{Acknowledgements}
The author would like to acknowledge the Minnesota Supercomputing Institute (MSI) at the University of Minnesota for providing computing facilities.

\bibliographystyle{plain}
\bibliography{ref}

\end{document}